\documentclass[11 pt]{amsart}

\usepackage{amssymb,color}
\usepackage{latexsym}
\usepackage{graphicx}
\usepackage{amsmath}
\usepackage{amsfonts}
\usepackage{natbib}

\newcommand{\cost}{\mathcal{COST}^n}
\newcommand{\s}{\mathcal{SUM}}
\newcommand{\tr}{\mathcal{TR}}
\newcommand{\si}{\mathcal{SC}}
\allowdisplaybreaks
\iftrue 
\usepackage{amsmath}
\usepackage{amsfonts}
\newcommand{\field}[1]{\mathbb{#1}}
\DeclareMathOperator{\PR}{\field{P}}             
\DeclareMathOperator{\E}{\field{E}}              
\def\N{\field{N}}                                
\def\R{\field{R}}                                
\def\F{\field{F}}                                

\else
\def\PR{\mathop{\rm I\kern -0.20em P}\nolimits}  
\def\E{\mathop{\rm I\kern -0.20em E}\nolimits}   
\def\N{\mathop{\rm I\kern -0.20em N}\nolimits}   
\def\R{\mathop{\rm I\kern -0.20em R}\nolimits}   
\def\F{\mathop{\rm I\kern -0.20em F}\nolimits}   
\fi

\newtheorem{thm}{Theorem}[section]

\newtheorem{lem}[thm]{Lemma}

\newtheorem{ex}[thm]{Example}

\numberwithin{equation}{section}

\title[ Warranty claims]{Modeling total expenditure on warranty claims} 

\author[ A.\ Mitra ]{Abhimanyu Mitra }
\address{Abhimanyu Mitra\\School of OR\&IE, Cornell University,
Ithaca, NY-14853} \email{am492@cornell.edu}

\author[ S. I.\ Resnick ]{Sidney I. Resnick}
\address{Sidney I. Resnick\\School of OR\&IE, Cornell University,
Ithaca, NY-14853} \email{sir1@cornell.edu}


\thanks{S. I. Resnick and A. Mitra were partially supported by ARO
  Contract W911NF-07-1-0078  at Cornell University.}

\normalsize

\begin{document}

\begin{center}
\maketitle

\end{center}
\normalsize

\begin{abstract}

 We approximate the distribution of total expenditure of a retail company over warranty claims incurred in a fixed period $[0, T]$, say the following quarter. We consider two kinds of warranty policies, namely, the non-renewing free replacement warranty policy and the non-renewing pro-rata warranty policy. Our approximation holds under modest assumptions on the distribution of the sales process of the warranted item and the nature of arrivals of warranty claims. We propose a method of using historical data to statistically estimate the parameters of the approximate distribution. Our methodology is applied to the warranty claims data from a large car manufacturer for a single car model and model year.

 \end{abstract}

\section{Introduction} \label{intro}

Suppose, a retail company sells items each of which is covered by a warranty for a period $W.$  So, the company estimates future warranty costs over a fixed period $[0,T],$ say the following quarter, based on historical data on sales and warranty claims. Typically, the length of the period for which we estimate total warranty cost, i.e. $T,$ if the period we are considering is $[0, T]$, is much smaller than the warranty period $W.$ For example, for a car company, usually the warranty period $W$ is three years whereas $T$ is a quarter. We assume $2T < W.$

We consider two kinds of warranty policies, namely, the non-renewing free replacement warranty policy and the non-renewing pro-rata warranty policy. Under the first policy, the retail company  agrees to repair or replace the item in case of a failure within the warranty period $W$. Under the second policy, the retail company refunds a fraction of the purchase price if the item fails within the warranty period $W$. The fraction depends on the lifetime of the item and is applicable to non-repairable items such as automobile tires; see \citet[page 171, 133]{blischke:murthy:1994}.

The role and importance of warranty costs in the retail industry has increased considerably and a considerable amount of research estimates warranty costs; see \cite{amato:anderson:1976, ja:kulkarni, kalbfleisch:lawless:robinson:1991, kulkarni:resnick:2007,  majeske:2003, sahin:polatoglu}. For the non-renewing free replacement policy case, under highly structured assumptions on the sales process and the times of claims, \citet{kulkarni:resnick:2007} found a closed form expression for the Laplace transform of the total warranty cost for a quarter, allowing computation of quantiles. Rather than attempting a closed form solution of the Laplace transform, we study approximations of the distribution of total warranty cost in a quarter under fairly modest assumptions on the distribution of the sales process of the warranted item and the nature of arrivals of warranty claims. Depending on the distribution of the cost of individual claims, we approximate the distribution of total warranty cost by a normal or a stable distribution. Computation of quantiles by our method is relatively straightforward.  In the case of companies issuing non-renewing pro-rata warranties, we approximate the distribution of the total warranty cost by a normal distribution.

The advantage in approximating total warranty cost using our asymptotic results is that our method does not require strong assumptions on the sales process distribution or on the nature of arrival of claims, and hence is robust against model error. In practice, the times of sales may not fit the renewal or Poisson process models; see Section \ref{cardata}. Similar problems are faced when modeling times of claims and here also our method based on asymptotic results provides an alternative by doing away with the strict assumptions on the distribution of times of claims.

We discuss methods of estimating the parameters of the normal or stable distribution, which approximate the distribution of the total warranty cost in $[0, T]$. We apply our methods to the sales and warranty claims data from a large car manufacturer for a single car model and model year.

\subsection{Outline}

The following sections are designed as follows. Section \ref{notations} reviews some notation. In Section \ref{freereplacement}, we discuss the case of non-renewing free replacement warranty policy. Section \ref{prorata} discusses the case of the non-renewing pro-rata policy. Both kinds of warranty policies use the same assumption on the distribution of the sales process. In Section \ref{salesproc}, we show that many common models for sales processes satisfy our assumptions. In Section \ref{dataandestimation}, we propose a method of estimating parameters of the approximate distribution of the total warranty cost in $[0, T]$. Section \ref{cardata} applies  our methods to the sales and warranty claims data from a large car manufacturer for a single car model and a single model year. The paper closes with some concluding remarks about the applicability of our results and possible future directions. The proofs of the main results are deferred to Section \ref{sec:proofs}.

\subsection{Notation}\label{notations}

The point measure on $K \subset \R$ corresponding to the point $x$ is given by $\epsilon_{x},$ i.e. for any Borel set $A \subset K,$
\begin{align*}
\epsilon_{x} (A) = \left\{\begin{array}{cc} 1 & \hbox{if $x \in A$, }\\ 0 & \rm{otherwise}.
\end{array}\right.
\end{align*}
The set of all Radon point measures on $K \subset \R$ is denoted $M_p(K)$. Similarly, the set of all non-negative Radon measures on $K \subset \R$ is denoted $M_+(K)$.

The set of right continuous functions with left limits from $[-W, T]$ to $\R$ is denoted $D([-W, T])$ and the set of continuous functions from $[-W, T]$ to $\R$ is denoted $C([-W, T])$ \citep[page 80, 121]{billingsley:1999}. Endow $D([-W, T])$ with the Skorohod topology and $C([-W, T])$ with the uniform topology. 

The set of all one dimensional regularly varying functions with exponent of variation $\rho$ is written $RV_{\rho}$ \citep[page 24]{resnickbook:2007}. Also, we denote conditional expectation of Y conditioned on $X$ as $E^X[Y],$ i.e. $E^X[Y] = E[Y | X].$ 

For easy reference, we give a glossary of notation in Section \ref{notation}.

\section{Non-renewable free replacement warranty policy}\label{freereplacement}

The free replacement policy is the most widely used warranty \cite[page 131]{blischke:murthy:1994} and is used for items such as cars, consumer electronics, etc. Under this warranty, the retail company repairs or replaces the item in case of a failure within the warranty period $W$ \cite[page 133]{blischke:murthy:1994}. Typically, such policies are non-renewing. 

{\it{Sales process}}: If a warranty claim for an item comes in the period $[0, T]$, the item must be sold during the period $[-W, T]$. As a setting for our approximation procedure, imagine a family of models indexed by $n.$ Let $S_j^n$ be the time of sale of the $j$-th item in the period $[-W, T]$. The sales process $N^n(\cdot)$ is the point process
\begin{align*}
N^n(t) = \sum_j \epsilon_{S_j^n}([-W, t]) = |\{ j : -W \le S_j^n \le t\} |.
\end{align*}
We further assume that $N^n(\cdot)$ is a random element of $D([-W, T])$ \citep[page 121]{billingsley:1999}.

We define a Gaussian process $\left(N^{\infty}(t), t \in [-W, T] \right)$ having continuous paths, so $N^{\infty}(\cdot)$ is also a random element of $C([-W, T])$ \citep[page 80]{billingsley:1999}. Existence of such a process can be guaranteed by the Kolmogorov's continuity theorem \cite[page 14]{Oksendal:2003}. We assume that the sales process $N^n(\cdot)$, after suitable scaling and centering, converges in distribution as $n \to \infty$ to the limiting process $N^{\infty}( \cdot)$ in $D([-W, T])$.

{\it{Times of claims measures}}: In the $n$-th model, let $C_{j, i}^n$ be the time of the $i$-th claim for the $j$-th item sold, where we start the clock at the time of sale $S_j^n$ of the $j$-th item,  so that $S_j^n + C_{j, i}^n$ is the actual claim time. Assume for all $j,$ the points $\{ C_{j, i}^n, i = 1, 2, \cdots \}$ do not cluster. The times of claims measure for the $j$-th item $M^n_j (\cdot)$ is 
\begin{align}\label{defn_mjn}
 M^n_j (A) = \sum_i \epsilon_{C_{j, i}^n} (A), \hskip 1 cm A \in \mathcal{B}([0, W]),
\end{align}
where $\mathcal{B}([0, W])$ is the set of all Borel subsets of $[0, W]$ and $ M^n_j (\cdot) \in M_p([0, W])$. Recall, $W$ is the warranty period and so only claims in $[0, W]$ will be respected.

In the $n$-th model, assume the random measures $\{M_j^n(\cdot), j \ge 1\}$ are independent and identically distributed for all $j$. Moreover, assume the common distribution of the random measures remain the same for all $n$. We denote the generic random measure describing claim times as $M(\cdot),$ i.e. $M(\cdot) \stackrel{d}{=} M_j^n(\cdot)$ for all $j$ and all $n$. 

{\it{ Claim sizes}}: We assume that the claim amounts are independent of the times of claims measures $\{M_j^n(\cdot): j \ge 1 \}$ and the sales process $N^n(\cdot)$ and claim amounts for different claims are independent and identically distributed. 

We consolidate detailed assumptions in the following section.

\subsection{Assumptions} \label{freereplass}
\begin{enumerate}
\item \label{freerepl_sale} Suppose, $\nu(\cdot)$ is a non-decreasing function in $D([-W, T])$ which is continuous at the points $T-W$ and $0.$ The family of centered and scaled sales processes in $[-W, T]$ converges weakly to a continuous path Gaussian process $N^{\infty} (\cdot)$ in $D([-W, T])$; i.e. 
\begin{equation}\label{eqn:salesprocess}
\sqrt{n}\left( \frac{N^n(\cdot)}{n} - \nu(\cdot) \right) \Rightarrow N^{\infty}(\cdot).
\end{equation}
Denote the mean function of $N^{\infty}(t)$ as $\theta(t) = E[N^{\infty}(t)]$ and the covariance function as $\gamma(s,t) = Cov[N^{\infty}(s), N^{\infty}(t)].$
\item \label{freerepl_claim1} For each $n$, the times of claims measures $\{M_j^n(\cdot), j \ge 1\}$ corresponding to different items sold are independent and identically distributed and the distribution of $\{ M_j^n(\cdot): j \ge 1\}$ remains the same for all $n$. The random measure $M(\cdot)$ denotes a random element of $M_p([0, W])$ whose distribution is the same as the common distribution of $\{M_j^n(\cdot) : j \ge 1, n \ge 1 \}$, i.e. $M(\cdot) \stackrel{d}{=} M_j^n(\cdot)$ for all $j$ and all $n$. For each $n$, the random measures $\{ M_j^n(\cdot): j \ge 1 \}$ are all assumed to be independent of the sales process $N^n(\cdot)$.
\item \label{freerepl_claim1.5} The random measure $M(\cdot)$ is a Radon measure with no fixed atoms except possibly at $0$ and $W,$ i.e. for $0 < x < W,$ $P[M(\{x \}) = 0] = 1.$
\item \label{freerepl_claim2} We assume $M(\cdot)$ satisfies $E[ M^2( [0, W] ) ] < \infty$.
\item \label{size_iid} For each $n$, the claim amounts for different claims are independent and identically distributed. The common distribution of  the claim sizes does not change with $n$.
\item \label{ind_amount_time} For each $n$, the claim amounts are independent of the times of claims measures $\{ M_j^n(\cdot): j \ge 1 \}$ and the sales process $N^n(\cdot)$.
\end{enumerate}

\subsection{Asymptotic approximation of total warranty cost distribution} In the $n$-th model, denote the total number of claims for the $j$-th item sold, that arrived in the fixed period $[0, T]$ by $R^n_j$:
\begin{equation}\label{define_rnj}
R_j^n = \sum_i \epsilon_{ S_j^n + C_{j, i}^n} ([0, T])\epsilon_{C_{j, i}^n} ([0, W]),
\end{equation}
and the total number of claims in $[0, T]$ as $R^n$:
\begin{equation}\label{define_rn}
R^n = \sum_{\{ j: -W \le S_j^n \le T \}} \sum_i \epsilon_{ S_j^n + C_{j, i}^n} ([0, T])\epsilon_{C_{j, i}^n} ([0, W]) = \sum_{\{ j: -W \le S_j^n \le T \}} R^n_j.
\end{equation}

We require some notation to state the results. Let $r : [0, W] \rightarrow [0,1]$ be a non-negative non-increasing function such that $r(0) = 1$. Recall the random measure $M(\cdot)$ defined in Assumption \ref{freerepl_claim1} of Section \ref{freereplass} and denote its expectation by $m(\cdot) = E[M(\cdot)].$ Then, $r(y)M(dy)$ is a random Radon measure on $[0, W]$ with expectation $\tilde m(\cdot)$, such that for all Borel sets $A$ of $[0, W],$ 
\begin{equation}\label{define_tilde_m}
\tilde m(A) = E\left[\int_A r(y)M(dy)\right].
\end{equation}

Now, define for $ x \in [-W, T],$
\begin{align}\label{define_tilde_f}
\delta(x) &= \left \{\begin{array}{ll}  \int_{[0, T-x] } r(y) M(dy), & \hbox{if $ 0 \le x \le T$},\\
 \int_{[-x, T-x] } r(y) M(dy), & \hbox{if $ T - W < x < 0$},\\
  \int_{[-x, W] } r(y) M(dy), & \hbox{if $ -W \le x \le T -W$}.\end{array} \right.
  \end{align}
Note that $\delta(\cdot)$ is a random function, whose interpretation depends on the kind of warranty policy. In the free replacement warranty policy, where $r \equiv 1$, $\delta(x)$ gives the number of claims in $[0, T]$ for an item sold at time $x$, i.e. $P[\delta(x) \in \cdot ] = P[R^n_1 \in \cdot | S^n_1 = x]$. The point-wise expectation and variance of $\delta(\cdot)$ are given by
  \begin{align}\label{define_f1}
  f_1(x)  = E [\delta (x)]   =  \left \{\begin{array}{ll} \tilde m([0, T-x]), & \hbox{if $ 0 \le x \le T$},\\
 \tilde m([-x, T-x]), & \hbox{if $ T - W < x < 0$},\\
  \tilde m([-x, W]), & \hbox{if $ -W \le x \le T -W$},\end{array} \right. 
\end{align}
and
\begin{align}\label{define_f2}
f_2(x) = Var[ \delta (x)].
\end{align}

Now, we define a function $\chi: D([-W, T]) \mapsto \R^{[0, W]}$ by 
\begin{align}\label{defineg}
  \chi(x)(u) = x(T-u) - x((-u)-), \hskip 1cm x \in D([-W, T]).
  \end{align}
Recall the Gaussian process $N^{\infty}(\cdot)$ in \eqref{eqn:salesprocess}. The Gaussian random variable $\int_{[0, W]} \chi(N^{\infty})(u) \tilde m(du)$ has expectation $\tilde \mu$  and variance $\tilde \sigma^2$ given by
\begin{align}\label{dotmudotsigma}
\tilde \mu &= \int_{[0, W]} E[ \chi \left(N^{\infty} \right)(u)] \tilde m(du) = \int_{[0, W]} E[ \chi \left(N^{\infty} \right)(u)] r(u) m(du),\nonumber \\
\tilde \sigma^2 &= \int_{[0, W]} \int_{[0, W]} Cov\left[ \chi \left(N^{\infty} \right)(u), \chi \left(N^{\infty} \right)(v) \right] \tilde m(du)\tilde m(dv) \nonumber \\
&= \int_{[0, W]} \int_{[0, W]} Cov\left[ \chi \left(N^{\infty} \right)(u), \chi \left(N^{\infty} \right)(v) \right] r(u)r(v) m(du) m(dv).
\end{align}

In the non-renewing free replacement policy, we choose $r(t) \equiv 1$. Hence, the measure $\tilde m$ defined in \eqref{define_tilde_m} coincides with $m(\cdot) = E[M(\cdot)].$ Similar simplifications occur in the definitions of $\delta$, $f_1$ and $f_2,$ as defined in \eqref{define_tilde_f}, \eqref{define_f1} and \eqref{define_f2} respectively.  The random function $\delta$ when $r \equiv 1$, is 
\begin{align}\label{defnsimpletildef}
\delta(x) &= \left \{\begin{array}{cl} 
M([0, T -x]), & \hbox{if $0 \le x \le T$},\\
M([ -x, T-x]), & \hbox{if $T- W < x < 0$},\\ 
M([ -x, W]), & \hbox{if $-W \le x \le T-W$},\end{array} \right.
\end{align}
and the expectation and variance are given by
  \begin{align}\label{define_simplef1}
  f_1(x)  = E [\delta(x)]   =  \left \{\begin{array}{ll} m([0, T-x]), & \hbox{if $ 0 \le x \le T$},\\
 m([-x, T-x]), & \hbox{if $ T - W < x < 0$},\\
 m([-x, W]), & \hbox{if $ -W \le x \le T -W$},\end{array} \right. 
\end{align}
and
\begin{align}\label{define_simplef2}
f_2(x) = Var[ \delta (x)].
\end{align}

From now on, till the end of Section \ref{freereplacement}, we use $\delta$, $f_1$ and $f_2$ to mean these simplified versions of them. We define two constants $c_1$ and $c_2$ as
\begin{align}\label{candc2freereplacement}
c_1 = \int_{[-W, T]} f_1(x)\nu(dx),\hskip 1 cm c_2 = \int_{[-W, T]}f_2(x) \nu(dx),
\end{align}
where $\nu(\cdot), f_1(\cdot)$ and $f_2(\cdot)$ are given in \eqref{eqn:salesprocess}, \eqref{define_simplef1} and \eqref{define_simplef2} respectively.
Also, since we chose $r(t) = 1$ for all $0 \le t \le W,$ the parameters $\tilde \mu$ and $\tilde\sigma^2$ defined in \eqref{dotmudotsigma} takes the simplified forms
\begin{align}\label{tildemutildesigma}
\tilde \mu &= \int_{[0, W]} E[ \chi \left(N^{\infty} \right)(u)] m(du),\nonumber \\
\tilde \sigma^2 &= \int_{[0, W]} \int_{[0, W]} Cov\left[ \chi \left(N^{\infty} \right)(u), \chi \left(N^{\infty} \right)(v) \right]m(du) m(dv).
\end{align}

\begin{thm} \label{numberclaims} 
\rm{Under Assumptions \ref{freerepl_sale}-\ref{freerepl_claim2} of Section \ref{freereplass}, the total number of claims $R^n$ is asymptotically normal; i.e. $\sqrt{n}\left( \frac{R^n}{n} - c_1 \right) \Rightarrow \mathcal{N}( \tilde \mu , c_2 + \tilde \sigma^2 )$, 
where $\mathcal{N}( a, b)$ is the normal distribution with mean $a$ and variance $b$, $c_1$ and $c_2$ are given in \eqref{candc2freereplacement} and $\tilde \mu$ and $\tilde \sigma^2$ are given in \eqref{tildemutildesigma}.
}
\end{thm}

Let, $\cost([0,T])$ be the total warranty cost during $[0,T]$ in the $n$-th model. Let $\{ X_i \}$ be iid with common distribution $F$ representing claim sizes in $[0, T]$. Denote, $\s_j = \sum_{i =1}^j X_i$ for all $j \ge 1$. Then, $\cost([0, T]) = \sum_{i = 1}^{R^n} X_i = \s_{R^n}$. 

The distribution $F$ of claim sizes is modeled as having a finite or infinite variance. Distributions having infinite variance are often assumed to have regularly varying tails \citep[page 344]{bingham:goldie:teugels:1987}. When $F$ has infinite variance, we assume $\bar F = 1- F \in RV_{-\alpha}, \hskip 0.2 cm 0 < \alpha < 2.$ 

The following theorem allows us to approximate the distribution of $\cost([0,T])$ based on the assumption we make about the claim size distribution $F$. 

\begin{thm}\label{mainthm}
\rm{
Under Assumptions \ref{freerepl_sale}-\ref{ind_amount_time} of Section \ref{freereplass}, we approximate the total cost as follows:
\begin{enumerate}
\item \label{normalthm}
 Suppose, the claim size distribution $F$ is such that $V = \int x^2F(dx) - {\left(\int xF(dx) \right)}^2 < \infty$. Then, as $n \to \infty$,
 \begin{equation}\label{normalapprox1}
 \frac{\cost([0,T]) - nc_1E}{\sqrt{n V} }\Rightarrow \mathcal{N}(\frac{E}{\sqrt{V}}\tilde \mu, c_1 + \frac{E^2}{V}(c_2 + \tilde \sigma^2) ),
 \end{equation}
 where $\mathcal{N}( \cdot, \cdot)$, $c_1$, $c_2$, $\tilde \mu$ and $\tilde \sigma^2$ are the same as in Theorem \ref{numberclaims} and $E = \int x F(dx)$.
 
 \item \label{stablemeanexists} Suppose, the claim size distribution $F$ is such that $\bar F (x) \in RV_{ - \alpha}$, $1 < \alpha < 2.$ Define, $b(x) = {\left(\frac{1}{1 -F}\right)}^{\leftarrow}(x)$.
Then, as $n \to \infty$,
 \begin{equation}\label{stableapprox2}
 \frac{\cost([0,T]) - nc_1E}{b(n)}\Rightarrow c_1^{\frac{1}{\alpha}}Z_{\alpha}(1),
 \end{equation}
 where $c_1$ is the same as in Theorem \ref{numberclaims}, $E = \int x F(dx)$ and $Z_{\alpha}(\cdot)$ is an $\alpha$-stable L\'{e}vy motion with $Z_{\alpha}(1)$ having characteristic function of the form
 \begin{align}\label{charstablemean}
 E\left[ \exp ( i \tau Z_{\alpha}(1) ) \right] = \exp \left( \int_0^{\infty} ( e^{i \tau x} - 1 - i \tau x) \alpha x^{-\alpha -1} dx \right).
 \end{align}
 
 \item \label{stablethm} Suppose, the claim size distribution $F$ is such that $\bar F (x) \in RV_{ - \alpha}$, $0 < \alpha \le 1$. Define, $b(x) = {\left(\frac{1}{1 -F}\right)}^{\leftarrow}(x)$ and $e(x) = \int_0^{b(x)} x F(dx)$. Then, as $n \to \infty$,
 \begin{equation}\label{stableapprox1}
 \frac{\cost([0,T]) - nc_1^{ \frac{1}{\alpha} }e(n)}{b(n)}\Rightarrow Z_{\alpha}(c_1) + 1_{\{ \alpha =1 \} } c_1\log c_1,
 \end{equation}
 where $c_1$ is the same as in Theorem \ref{numberclaims} and $Z_{\alpha}(\cdot)$ is an $\alpha$-stable L\'{e}vy process with $Z_{\alpha}(c_1)$ having characteristic function of the form
 \begin{align}\label{charstablenomean}
 E\left[ \exp ( i \tau Z_{\alpha}(c_1) ) \right] = \exp \left[ c_1 \left( \int_1^{\infty} ( e^{i \tau x} - 1) \alpha x^{-\alpha -1} dx + \int_0^1 ( e^{i \tau x} - 1 - i \tau x) \alpha x^{-\alpha -1} dx \right)\right].
 \end{align}
 
 \end{enumerate}
 }
\end{thm}

\section{Non-renewable pro-rata warranty policy}\label{prorata}

The non-renewable pro-rata warranty policy is commonly used for consumer durables such as automobile batteries and tires \cite[page 169]{blischke:murthy:1994}. Under this policy, the manufacturer pays a fraction of the cost of the item in case of failure within the warranty period $W$. The fraction depends on the lifetime of the item. So, if an item of cost $c_b$ fails after time $t$ from the date of purchase, the manufacturer pays the amount $q(t)$, where
\begin{equation}\label{rebatedefn}
q(t) = \left\{ \begin{array}{cc} c_b r(t) & \hbox{if $t \le W$,}\\
0 & \hbox{otherwise,} \end{array}\right.
\end{equation} 
where $r : [0, W] \rightarrow [0,1]$ is a non-negative decreasing function with $r(0) = 1.$ We call the function $r(\cdot)$ our rebate function. In many situations, the rebate function is taken to be a linear or quadratic function of the lifetime of the item; see \citet[page 172]{blischke:murthy:1994}.

In this section, since there is no repair or replacement, each item sold can have at most one warranty claim. So, the times of claims measure $M(\cdot)$ has the additional property that for any Borel measurable set $A \subset [-W, T]$, $M(A)$ can only assume two values, $0$ or $1$ and Assumption \ref{freerepl_claim2} of Section \ref{freereplass} is always satisfied.

\subsection{Approximation of the distribution of total cost of warranty claims} As before, let  $\cost([0, T])$ be the total expenditure  on warranty claims during the fixed period $[0,T]$ in the $n$-th model. Let $C^n_{j, 1}$ be the lifetime of the $j$-th item sold. Then,
\begin{align*}
\cost([0, T]) &= \sum_{ \{j : S_j^n \in [-W, T] \}}  c_br(C^n_{j, 1})\epsilon_{C_{j, 1}^n}([0, W])\epsilon_{ S_j^n + C_{j, 1}^n} ([0, T]).
\end{align*}
Recall the definitions of $\tilde m, \delta, f_1, f_2$ and $\chi$ given in \eqref{define_tilde_m}, \eqref{define_tilde_f}, \eqref{define_f1}, \eqref{define_f2} and \eqref{defineg} respectively. In the case of pro-rata warranty policy, the function $r(\cdot)$ in the definition of the random function $\delta(\cdot)$ given in \eqref{define_tilde_f} is the same as the rebate function $r(\cdot)$ defined in \eqref{rebatedefn}. Here the random function $\delta(x)$ is interpreted as the proportion of price spent on warranty claims for an item sold at time $x$. Note that the rebate function $r(\cdot)$ is known and the randomness of $\delta(\cdot)$ stems solely from the random measure $M(\cdot)$. We denote the mean and variance of the Gaussian random variable $\int_{-W}^T \chi (N^{\infty})(u)\tilde m(du)$ by $\tilde \mu$ and $\tilde\sigma^2$ respectively. The forms of $\tilde \mu$ and $\tilde \sigma^2$ are given in \eqref{dotmudotsigma}. We also define two constants $c_1$ and $c_2$ as
\begin{align}\label{candc2prorata}
c_1 = \int_{[-W, T]} f_1(x)\nu(dx),\hskip 1 cm c_2 = \int_{[-W, T]}f_2(x) \nu(dx),
\end{align}
where $\nu(\cdot), f_1(\cdot)$ and $f_2(\cdot)$ are given in \eqref{eqn:salesprocess}, \eqref{define_f1} and \eqref{define_f2} respectively.

\begin{thm}\label{proratathm}
\rm{In case of pro-rata warranty policy, under the Assumptions \ref{freerepl_sale}-\ref{freerepl_claim1.5} of Section \ref{freereplass}, 
\begin{equation*}
\frac{\cost([0, T]) - n c_bc_1}{c_b\sqrt{n}} \Rightarrow \mathcal{N} ( \tilde \mu, c_2 + \tilde \sigma^2)
\end{equation*}
where $c_b$ is the price of each warranted item, $c_1$ and $c_2$ are given in \eqref{candc2prorata} and  $\tilde \mu$ and $\tilde \sigma^2$ are given in \eqref{dotmudotsigma}.
}
\end{thm}

\section{Examples of sales processes}\label{salesproc}

In the earlier two sections, we have considered the free replacement warranty policy and the pro-rata warranty policy. In both cases, the assumptions on the distribution of the times of claims measure $M(\cdot)$ as given in Assumptions \ref{freerepl_claim1.5} and \ref{freerepl_claim2} of Section \ref{freereplass} are modest and a vast class of measures qualify. In comparison, the assumption on the sales process $N^n(\cdot)$ given in \eqref{eqn:salesprocess} is stricter. Here we list several sales processes satisfying \eqref{eqn:salesprocess}.

\begin{ex}\label{renewal}
\rm{ {\bf{Renewal Processes}}\\
Suppose, $N(\cdot)$ is a renewal process on $[0, \infty),$ where the common inter-arrival distribution has mean $\phi_1$ and variance $\phi_2$. For the $n$-th model, define the sales process $N^n(\cdot)$ as
$N^n(s) = N(n(s +W))$ for $s \in [-W, T]$ and define $B(\cdot)$ to be the Brownian motion on $[0, \infty)$. Then, from (9.4) of \citet[page 293]{resnickbook:2007} or Theorem 14.6 of \cite[page 154]{billingsley:1999}, we get
\begin{align}\label{renewaleq}
\sqrt{n} \left( \frac{1}{n}N(n(s + W)) - \frac{s +W }{\alpha} \right) \Rightarrow \frac{\sqrt{\phi_2}}{\phi_1^{3/2}}B(s + W)
\end{align}
on $D([-W, T])$. Define, $\nu(s) = \frac{s +W}{\alpha}$ and $N^{\infty}(s) = \frac{\sqrt{\phi_2}}{\phi_1^{3/2}}B(s +W)$. The homogeneous Poisson process is a special case. 
}
\end{ex}

\begin{ex}\label{poisson}

\rm{ {\bf{Non-homogeneous Poisson Processes}}\\
Suppose, $\nu : [-W, T] \rightarrow [0, \infty)$ is a continuous strictly increasing function and $N(\cdot)$ a homogeneous Poisson process on $[0, \infty)$ with intensity 1. Now, define the sales process $N^n(\cdot)$ as $N^n(\cdot) = N(n\nu(\cdot))$ and define $B(\cdot)$ to be a Brownian motion on $[0, \infty).$ Applying (9.4) of \citet[page 293]{resnickbook:2007} in the case of $N(\cdot)$, we get
 \begin{align}\label{poissoneq}
\sqrt{n} \left( \frac{1}{n}N(ns) - s \right) \Rightarrow B(s)
\end{align}
on $D([0, \infty))$. Define the composition function $\psi : D([0, \infty)) \rightarrow D([-W,T])$ by $\psi(x) = x \circ \nu$, and since $\psi(\cdot)$ is continuous \citep[Theorem 3.1]{wardwhitt}, using the continuous mapping theorem \cite[page 21]{billingsley:1999} to \eqref{poissoneq}, we get
 \begin{align*}
\psi \left( \sqrt{n} \left( \frac{1}{n}N(ns) - s \right)\right) \Rightarrow \psi \left(B(s)\right)
\intertext{on $D([-W, T]),$ which implies}
\sqrt{n} \left( \frac{1}{n}N(n\nu(s)) - \nu(s) \right) \Rightarrow B(\nu(s))
\end{align*}
on $D([-W, T])$. Define, $N^{\infty}(\cdot) =  B(\nu(\cdot))$ and Assumption \ref{freerepl_sale} of Section \ref{freereplass} holds.
}
\end{ex}

\begin{ex}
\rm{ {\bf{Doubly Stochastic Poisson Processes}}\\
Define $\nu(\cdot)$, $B(\cdot)$ and $N(\cdot)$ as in Example \ref{poisson} and let $D_0$ be the subset of non-negative non-decreasing functions of $D([-W, T])$. Assume, there exists a sequence of random elements $\{ \Lambda^n \}$ of $D_0$ independent of $N(\cdot)$ and after centering and scaling the sequence converges to a continuous Gaussian process $N_2^{\infty}(\cdot)$ in $D([-W, T])$; i.e. 
\begin{align}\label{doubstoceq1}
\frac{ \Lambda^n( \cdot) - n \nu(\cdot) }{\sqrt{n}} \Rightarrow N_2^{\infty} (\cdot)
\end{align}
on $D([-W, T]).$ Now, define the sales process $N^n(\cdot) = N( \Lambda^n(\cdot))$. Using the fact that $N(\cdot)$ is independent of $\{ \Lambda^n \}$, \eqref{poissoneq} and \eqref{doubstoceq1} yield \citep[page 25]{billingsley:1999}
\begin{align*}
\left( \begin{array}{c} \sqrt{n} \left( \frac{1}{n}N(n\cdot) - (\cdot) \right) \\ \frac{1}{\sqrt{n}}(\Lambda^n( \cdot) - n \nu(\cdot) ) \end{array} \right) \Rightarrow \left( \begin{array}{c} B(\cdot) \\ N_2^{\infty} (\cdot)
\end{array} \right)
\intertext{on $D([-W, T]) \times D([-W, T])$, where $N_2^{\infty} (\cdot)$ and $B(\cdot)$ are independent of each other. Further, using \citet[page 37]{billingsley:1999} and $\Lambda^n(\cdot)/n \Rightarrow \nu(\cdot)$ gives}
\left( \begin{array}{c} \sqrt{n} \left( \frac{1}{n}N(n\cdot) - (\cdot) \right) \\ \frac{1}{n} \Lambda^n( \cdot) \\ \frac{1}{\sqrt{n}}(\Lambda^n( \cdot) - n \nu(\cdot) ) \end{array} \right) \Rightarrow \left( \begin{array}{c} B(\cdot) \\ \nu(\cdot) \\ N_2^{\infty} (\cdot)
\end{array} \right)
\intertext{on $D([-W, T]) \times D_0 \times D([-W, T])$ and from the continuous mapping theorem \cite[page 21]{billingsley:1999} and Theorem 3.1 of  \cite{wardwhitt}, we get}
\left( \begin{array}{c} \frac{1}{\sqrt{n}}\left( N(\Lambda^n(\cdot)) - \Lambda^n(\cdot) \right) \\ \frac{1}{\sqrt{n}}(\Lambda^n( \cdot) - n \nu(\cdot) ) \end{array} \right) \Rightarrow \left( \begin{array}{c} B(\nu(\cdot)) \\ N_2^{\infty} (\cdot)
\end{array} \right)
\intertext{ on $D([-W, T]) \times D([-W, T])$. Therefore, applying the addition functional,}
\sqrt{n} \left( \frac{1}{n}N(\Lambda^n(\cdot)) - \nu(\cdot) \right) \Rightarrow B(\nu(\cdot)) + N_2^{\infty}(\cdot)
\end{align*}
on $D([-W, T])$ and the processes $B(\cdot)$ and $N_2^{\infty} (\cdot)$ are independent of each other. With $N^{\infty}(\cdot) =  B(\nu(\cdot)) + N_2^{\infty}(\cdot)$, this model satisfies Assumption \ref{freerepl_sale} of Section \ref{freereplass}.

Assumption \eqref{doubstoceq1} is modest and Examples \ref{renewal} or \ref{poisson} satisfy \eqref{doubstoceq1}.
}
\end{ex}

\section{Estimation procedure}\label{dataandestimation}

An important estimation question is the choice of $n$. We interpret $n$ as a measure of the volume of sales of the warranted item. So, $n$ should depend on the size of the company and the nature of the warranted item. For example, we would expect larger $n$ for an ordinary car model than a luxury car model. We assume that for the time period we are considering, say $[-W, T]$, $n$ does not change. The non-stationarity of the sales process $N^n(\cdot)$ in this period is captured by the functions $\nu(\cdot), \theta(\cdot)$ and $\gamma(\cdot, \cdot)$ given in Assumption \ref{freerepl_sale} of Section \ref{freereplass}. If we are ambitious enough to predict the warranty cost on some time period further in future, say $[T, 2T]$, we will assume that $n$ does not change for the entire time period $[-W, 2T]$. Thus, we assume $n$ does not change for the entire time period we consider. Since we interpret $n$ as a measure of the sales volume, and $n$ does not change for the entire time period, we choose total sales in our observed sales data, say total sales in the time period $[-W, 0]$, for $n$.

We discuss estimation methods for both the non-renewing free replacement warranty policy and the non-renewing pro-rata warranty policy. 

\subsection{Free replacement policy} Which version of Theorem \ref{mainthm} should we apply: \eqref{normalthm}, \eqref{stablemeanexists} or \eqref{stablethm}? The answer depends on the data of claim sizes. We assumed claim sizes are i.i.d. with common distribution function $F$. A diagnostic for determining whether data comes from a heavy-tailed distribution is the QQ plot \cite[page 97]{resnickbook:2007}. If $\bar F = 1 - F \in RV_{-\alpha}$ for some $\alpha > 0$, we expect the QQ plot to be a straight line with slope $\frac{1}{\alpha}$. If we decide $\bar F \in RV_{-\alpha}$, we estimate $\alpha$ using one of the various estimators of $\alpha$ available in the literature \citep[Chapter 4]{resnickbook:2007}. Depending on the value of our estimate of $\alpha$, we determine which version of Theorem \ref{mainthm} to use. If our analysis yields that $\bar F \notin RV_{-\alpha}$, we verify that $F$ has finite variance and use version \eqref{normalthm} of Theorem \ref{mainthm}.

For versions \eqref{normalthm}, \eqref{stablemeanexists} and \eqref{stablethm} of Theorem \ref{mainthm}, the limit relations in \eqref{normalapprox1}, \eqref{stableapprox2} and \eqref{stableapprox1} have different sets of parameters. We proceed case by case to discuss how we estimate parameters in each case.
 
\subsubsection{Estimation of the parameters in the limit relation in Theorem \ref{mainthm}, version \eqref{normalthm}} \label{sec:estimationnormalthm}
We estimate six parameters given in \eqref{normalapprox1}: $c_1, c_2, \tilde \mu, \tilde \sigma^2, E$ and $V$. We estimate $E$ by the sample mean and $V$ by the sample variance of the claim sizes. 

For the rest of the parameters, we first analyze the sales data and estimate the functions $\nu(\cdot), \theta(\cdot)$ and $\gamma(\cdot, \cdot),$ given in  Assumption \ref{freerepl_sale} of Section \ref{freereplass}. We assume that we have observed sales for the period $[-W, 0]$ and have not observed sales for the period $[0, T]$.

One parametric approach for estimating $\nu(\cdot)$, which is adopted in Section \ref{cardata}, assumes that $n \nu(\cdot)$ follows the Bass model \citep{bass:1969} in the time period where we have observed sales, say $[-W, 0]$. Since the Bass model describes the pattern of sales from the introduction of an item in the market \citep{bass:1969}, this approach gets additional justification when we have sales data of the warranted item starting from its introduction in the market. The Bass model for total sales by time $t$, $T(t)$ (adjusted for our clock, since we have sales data for the period $[-W, 0]$) is given by
\begin{align*}
T(t) = n \frac{1 - \exp(-C(t+W))}{1 + (C/B - 1)\exp (-C(t+W))},
\end{align*} 
 where $n$ is the total sales in the time period of observed sales, say $[-W, 0]$. Hence, using the Bass model for $n \nu(\cdot),$ we get that $\nu(\cdot)$ must have the form $T(t)/n$ and to estimate $\nu(\cdot)$, we have to estimate the parameters $B$ and $C$. Let $\nu'(t)$ be the density of $\nu(\cdot)$ at $t$. We minimize the squared error
\begin{equation*}
\min_{B, C} \sum_{t = -W +1}^0 {\left[N^n(t) - N^n(t-1) - n\nu'(t) \right]}^2
\end{equation*}
to obtain estimates $(\hat B, \hat C)$. Using this procedure, we fit the Bass model to our observed data on sales (say, on the time period $[-W, 0]$) and then extrapolate $\nu(\cdot)$ on some future time period, say $[0, T]$, on which we have no sales data. We denote estimated $\nu(\cdot)$ as $\hat \nu (\cdot)$. Our estimation of $\nu(\cdot)$ is free from any distributional assumption on the sales process $N^n(\cdot)$.
 
 Now, we obtain the residuals $\{ r_t = n^{-1/2}\left( N^n(t) - N^n(t-1) - \hat \nu(t) + \hat \nu (t-1) \right): t = -W+1, -W+2, \cdots, 0 \}$. These residuals act as surrogates for $\{ N^{\infty}(t) - N^{\infty}(t-1): t = -W+1, -W+2, \cdots, 0\}$ (recall the limit relation in \eqref{eqn:salesprocess}). We use standard time-series techniques on $\{ r_t: t = -W+1, \cdots, 0 \}$ to get estimates of $\{ \hat \theta(t), \hat \gamma(t, s): t, s = -W+1, \cdots, 0\}.$
 
 We assume  that $N^{\infty}(t) - N^{\infty}(t-1) = \tr_t + \si_t Z(t)$, where $Z(t)$ is a stationary Gaussian process and $\si_t$ is function of $t$ which takes only positive values. Note that this is an additional assumption we need for estimation purposes. We have not assumed $E[Z(t)] = 0$ or $Var[Z(t)] = 1$.
 
 We first plot the time plot of $\{ r_t: t = -W+1, \cdots, 0 \}$. If the time plot looks stationary, we are done and assume $\tr_t \equiv 0$ and $\si_t \equiv 1$. Otherwise, we estimate $\tr_t$ and $\si_t$. We do moving average smoothing on $r_t$ to get $\hat \tr_t$, which estimates the trend. We plot absolute values of $(r_t - \hat \tr_t )$ and fit another moving average estimator to it to get $\hat \si_t$. 
 
 We assume $\{ j_t = (r_t - \hat \tr_t)/\hat \si_t: t = -W+1, \cdots, 0 \}$ act as surrogates for the stationary process $\{ Z(t): t = -W+1, \cdots, 0 \}$. We estimate the sample mean $l$, sample variance $s^2$ and sample autocorrelation function $c(\cdot)$ of $\{j_t \}$. Hence, $\{ \theta(t) - \theta(t-1): t = -W+1, \cdots, 0 \}$ is estimated as
 \begin{align*}
 \hat \theta(t) - \hat \theta (t-1) = \hat \tr_t + l \hat \si_t,
 \end{align*}
 and recover $\{ \hat \theta(t): t = -W+1, \cdots, 0 \}$. Similarly, $\{Cov[N^{\infty}(t)- N^{\infty}(t-1), N^{\infty}(s)- N^{\infty}(s-1)]: t = -W+1, \cdots, 0 \}$ is estimated as
  \begin{align} \label{covstructure}
\hat {Cov}[ N^{\infty}(t)- N^{\infty}(t-1), N^{\infty}(s)- N^{\infty}(s-1)] = \hat \si_t \hat \si_s s^2 c(t-s). 
 \end{align}
 From \eqref{covstructure}, $\{ \hat \gamma(\cdot, \cdot): t, s = -W+1, \cdots, 0 \}$ can be computed.
 
 We also require $\{ \hat \theta(t), \hat \gamma(t, s): t \in [0, T], s \in [-W, T] \}$. The problem in estimating $\{ \hat \theta(t), \hat \gamma(t, s): t \in [0, T], s \in [-W, T] \}$ is that we do not yet have estimates of $\{ \hat \tr_t, \hat \si_t, c(s) : t \in [0, T], \hskip 0.1 cm s \in [W, T+W] \}$. To get estimates of $\{ \hat \tr_t, \hat \si_t, c(s) : t \in [0, T], \hskip 0.1 cm s \in [W, T+W] \}$, fit a polynomial to both $\{ \hat \tr_t: t = -W+1, \cdots, 0 \}$ and $\{ \log (\hat \si_t): t = -W+1, \cdots, 0 \}$. We use the fitted polynomial values to estimate $\{ \hat \tr_t, \hat \si_t: t \in [0, T] \}$. We also assume $c(t) = 0$ if $t > W$, since we  only have data on sales from $[-W, 0]$. If we have sales data for a longer period, then it is also possible to estimate $c(t)$ for $t > W.$ Then, using estimates of $\{ \hat \tr_t, \hat \si_t, c(s) : t \in [0, T], \hskip 0.1 cm s \in [W, T+W] \}$ we obtain estimates of $\{ \hat \theta(t), \hat \gamma(t, s): t \in [0, T], \hskip 0.1 cm s \in [-W, T] \}$ following a similar procedure as the one used to obtain  $\{ \hat \theta(t), \hat \gamma(t, s): t, s = -W+1, \cdots, 0\}$. Thus, we complete our estimation of $\{\hat \theta(t), \hat \gamma(t, s) : t, s = -W, \cdots, T\}$.
 
Now analyze the warranty claims data to get an estimate for the distribution of the times of claims measure $M(\cdot)$, given in Assumption \ref{freerepl_claim1} of Section \ref{freereplass}. Recall that the times of claims measure in the $n$-th model for the $j$-th item sold is $M^n_j(\cdot)$. Also, by Assumption \ref{freerepl_claim1} of Section \ref{freereplass}, $\{ M^n_j(\cdot), j = 1, 2, \cdots, n \}$ are independent and identically distributed with common distribution as that of $M(\cdot)$. For each item $j$ in our sales data, we consider its times of claims  measure $M^n_j(\cdot).$ If an item $j$ has no record of claims, then we assume that $M^n_j \equiv 0$. We compute $\{ M^n_j((x-1, x]): x = 0, 1, \cdots, W \}$ with the interpretation that for $x = 0$, $M^n_j((x-1, x]) \equiv M_j^n(\{0\})$. From the plot of  $\left\{ \left( x, \frac{1}{n} \sum_{j = 1}^n M_j^n((x-1, x]) \right): x = 0, 1, \cdots, W \right\}$, we infer a functional form of the mean measure $m(\cdot) = E[M(\cdot)] = E[M^n_1(\cdot)]$. Getting a functional form of $m(\cdot)$ is useful because to compute $\tilde \mu$ and $\tilde \sigma^2$, given in \eqref{tildemutildesigma}, we have to integrate with respect to $m(dx)$; see Section \ref{sec:examplearrivalclaim} for an example. We denote estimated $m(\cdot)$ as $\hat m(\cdot)$.

Recall the definition of the parameters $c_1$ and $c_2$ given in \eqref{candc2freereplacement}. To estimate $c_1$ and $c_2$, we need to estimate first the functions $\{ \hat f_1(x) : x \in [-W, T] \}$ and $\{ \hat f_2(x) : x \in [-W, T] \}$. Actually, we estimate $\{ \hat f_1(x) : x = -W, -W +1, \cdots, T \}$ and $\{ \hat f_2(x) : x = -W, -W +1, \cdots, T \}$, and get estimates $ \hat c_1 = \int_{-W}^T \hat  f_1(x) \hat \nu(dx)$ and $\hat c_2 = \int_{-W}^T \hat f_2(x) \hat \nu(dx)$ using the trapezoid method of integration. We estimate $\{ \hat f_1(x) : x = -W, \cdots, T \}$ and $\{ \hat f_2(x) : x = -W, \cdots, T \}$ as  
\begin{equation}\label{f1estimation}
\hat f_1(x) = \left \{\begin{array}{ll} \hat m ([0, T-x]), & \hbox{if $ 0 \le x \le T$},\\
 \hat m([-x, T-x]), & \hbox{if $ T - W < x < 0$},\\
 \hat m([-x, W]), & \hbox{if $ -W \le x \le T -W,$} \end{array} \right.
 \end{equation}
 and
 \begin{equation}\label{f2estimation}
 \hat f_2(x) = \left \{\begin{array}{ll} \frac{1}{n} \sum_{j=1}^n {\left[M^n_j([0, T-x])\right]}^2 - {\left[ \hat f_1(x)\right]}^2 & \hbox{if $ 0 \le x \le T$},\\
  \frac{1}{n} \sum_{j=1}^n {\left[M^n_j([-x, T-x])\right]}^2 - {\left[ \hat f_1(x)\right]}^2, & \hbox{if $ T - W < x < 0$},\\
  \frac{1}{n} \sum_{j=1}^n {\left[M^n_j([-x, W])\right]}^2 - {\left[ \hat f_1(x)\right]}^2, & \hbox{if $ -W \le x \le T -W,$} \end{array} \right.\\
\end{equation}
where $n$ is the total number of items sold and $M^n_j(\cdot)$ is the times of claims measure for the $j$-th item sold in the $n$-th model.

Now, we are left with the estimation of $\tilde \mu$ and $\tilde \sigma^2$, given in \eqref{tildemutildesigma}. To estimate $\tilde \mu$ and $\tilde \sigma^2$, first we must estimate $\{ E[ \chi(N^{\infty})(u)]: u \in [0,W] \}$
and $\{ Cov[ \chi(N^{\infty})(u), \chi(N^{\infty})(v) ]: u, v \in [0,W] \}$, where $N^{\infty}(\cdot)$ is given in \eqref{eqn:salesprocess} and $\chi(\cdot)$ is defined in \eqref{defineg}. We estimate $\{ E[ \chi(N^{\infty})(u)], u = 0, 1, \cdots, W \}$ and $\{Cov[ \chi(N^{\infty})(u), \chi(N^{\infty})(v) ]: u, v = 0, 1, \cdots, W \}$ from the estimates of $\{ \hat \theta(t), \hat \gamma(t, s) : t, s = -W+1, \cdots, T \}$ as
\begin{align*}
{\hat E[ \chi(N^{\infty})(u)]} &= \hat \theta(T-u) - \hat \theta(u),\\
\intertext{and}
\hat {Cov}[ \chi(N^{\infty}(u)), \chi(N^{\infty}(v)) ] &= \hat \gamma (T-u, T-v) + \hat \gamma(-u, -v) -  \hat \gamma(T-u, -v) - \hat \gamma(T-v, -u),
\end{align*}
where $\{ \hat \theta(t), \hat \gamma(t, s) : t, s = -W+1, \cdots, T \}$ are estimates of $\{ \theta(t), \gamma(t, s) : t, s = -W+1, \cdots, T \}$ obtained while analyzing the sales process. The definitions of the functions $\theta(\cdot)$ and $\gamma(\cdot, \cdot)$ can be found in Assumption \ref{freerepl_sale} of Section \ref{freereplass}. Now, we integrate by the trapezoid method to obtain the estimated mean 
$$\hat {\tilde \mu} = \int_{[0, W]} \hat E[ \chi \left(N^{\infty} \right)(u)] \hat m(du)$$ 
and the estimated variance 
$$\hat {\tilde \sigma}^2 = \int_{[0, W]} \int_{[0, W]} \hat {Cov} \left[ \chi \left(N^{\infty} \right)(u), \chi \left(N^{\infty} \right)(v) \right] \hat m(du) \hat m(dv).$$ 

This method of estimation is applied to the sales and claims data of a car manufacturer for a specific model and model year in Section \ref{cardata}.

\subsubsection{Estimation of the parameters in the limit relation in Theorem \ref{mainthm}, version \eqref{stablemeanexists}} \label{sec:estimationstablemeanexists}

We estimate the parameters $c_1, E, \alpha, b(n)$ and the parameters of the stable distribution of $Z_{\alpha}(1)$, where $Z_{\alpha}(1)$ is given in \eqref{stableapprox2}. Estimate $c_1$ and $E$ in the same manner as described in Section \ref{sec:estimationnormalthm}. We estimate $\alpha$ by one of its estimators \cite[Chapter 4]{resnickbook:2007}, say the QQ-estimator. There are two ways to estimate $b(n)$: 
\begin{enumerate}
\item Use the $(1 - \frac{1}{n})$-th quantile of the iid data on claim sizes as $b(n)$; or 
\item Assume the claim size distribution $F$ is close to pareto and use $n^{1/\alpha}$ as an estimate of $b(n)$.
\end{enumerate}
 We adopt the second method of estimating $b(n)$ when analyzing data in Section \ref{cardata}.

For the stable distribution of $Z_{\alpha}(1)$, we follow the parameterization of  \citet[page 5]{samorodnitsky:taqqu:1994}. From \eqref{charstablemean}, we get that the parameters of the distribution of $Z_{\alpha}(1)$ are
\cite[page 171]{samorodnitsky:taqqu:1994}: 
\begin{align}\label{stablemeanpara}
\mu = 0, \hskip 0.5 cm \sigma = {\left(- \frac{\Gamma(2-\alpha)}{\alpha -1} \cos(\frac{\pi \alpha}{2}) \right)}^{\frac{1}{\alpha}}, \hskip 0.5 cm \beta = 1.
\end{align}
Obtaining an estimate of $\sigma$ from our estimate of $\alpha$ is a simple numerical procedure.

\subsubsection{Estimation of the parameters in the limit relation in Theorem \ref{mainthm}, version \eqref{stablethm} } Estimate $\alpha$, say using the QQ estimator. Depending on whether $0 < \alpha < 1$ or $\alpha = 1$, our estimators of parameters will be different, but in both cases, we have to estimate the same set of parameters: $c_1, e(n), b(n)$ and the parameters of the stable distribution of $Z_{\alpha}(c_1)$, where $Z_{\alpha}(c_1)$ is given in \eqref{stableapprox1}. Estimate $c_1$ using the same procedure discussed in Section \ref{sec:estimationnormalthm}.

When $0 < \alpha < 1$, we assume that the claim size distribution $F$ is quite close to Pareto and hence use $n^{1/\alpha}$ as an estimate of $b(n)$ and $\frac{\alpha}{1 - \alpha} \left( n^{(1 - \alpha)/\alpha} -1 \right)$ as an estimate of $e(n)$. For the stable distribution of $Z_{\alpha}(c_1)$, we follow the parameterization of  \citet[page 5]{samorodnitsky:taqqu:1994}. From \eqref{charstablenomean}, we get that the parameters of the distribution of $Z_{\alpha}(c_1)$ are \cite[page 170]{samorodnitsky:taqqu:1994}:

\begin{align*}
\mu = -\frac{c_1 \alpha}{1 - \alpha}, \hskip 0.5 cm \sigma = {\left(c_1\Gamma(1-\alpha) \cos(\frac{\pi \alpha}{2}) \right)}^{\frac{1}{\alpha}}, \hskip 0.5 cm \beta = 1.
\end{align*}
Computing estimates of $\mu$ and $\sigma$ using our estimates of $\alpha$ and $c_1$ is routine.

If $\alpha = 1$, we assume again that the claim size distribution $F$ is quite close to Pareto and hence use $n$ as an estimate of $b(n)$ and $\log n$ as an estimate of $e(n)$. For the stable distribution of $Z_{\alpha}(c_1)$, we follow the parameterization of  \citet[page 5]{samorodnitsky:taqqu:1994}. From \eqref{charstablenomean}, we get that the parameters of the distribution of $Z_{\alpha}(c_1)$ are \cite[page 166]{samorodnitsky:taqqu:1994}:

$$\mu = c_1 \int_0^{\infty} [ \sin z - z 1_{ \{z \le 1\} } ] z^{-2} dz , \hskip 0.5 cm \sigma = \frac{c_1\pi}{2}, \hskip 0.5 cm \beta = 1.$$
Computing estimates of $\mu$ and $\sigma$ using our estimates of $\alpha$ and $c_1$ is then routine.

\subsection{Pro-rata policy case}
 The estimation method in this case is mostly similar to the one described in Section \ref{sec:estimationnormalthm}. We need to estimate four parameters given in Theorem \ref{proratathm}: $c_1, c_2, \tilde \mu$ and $\tilde \sigma^2.$

First, observe that in this case, we do not need any data on claim sizes. Given the times of claims measures $\{M^n_j(\cdot): j = 1, 2, \cdots .  \},$ the claim sizes are determined by the function $q(\cdot)$ given in \eqref{rebatedefn}. 

We analyze the sales process in the same manner as described in Section \ref{sec:estimationnormalthm}. Thus, we obtain estimates of the mean and covariance functions of $\{\chi(N^{\infty})(u): u = 0, 1, \cdots, W \}$ given by $\{ E[ \chi (N^{\infty})(u)]: u = 0, \cdots, W\}$ and $\{ Cov[ \chi(N^{\infty}(u)), \chi(N^{\infty}(v)) ]: u, v = 0, \cdots, W \}$. We also estimate the mean times of claims measure $m(\cdot) = E[ M(\cdot)]$ following the same methods described in Section \ref{sec:estimationnormalthm}. We denote the estimate of $m(\cdot)$ as $\hat m(\cdot)$.

Now, recall the parameters $\tilde \mu$ and $\tilde \sigma^2$ given in \eqref{dotmudotsigma} and the rebate function $r(\cdot)$ defined in \eqref{rebatedefn}. Following \eqref{dotmudotsigma}, we estimate $\tilde \mu$ and $\tilde \sigma^2$ as 
\begin{align*}
\hat {\tilde \mu} = \int_{[0, W]} \hat E[ \chi\left(N^{\infty} \right)(u)] \hat {\tilde m}(du) = \int_{[0, W]} \hat E[ \chi \left(N^{\infty} \right)(u)]r(u) \hat m(du)
\end{align*}
 and 
 \begin{align*}
 \hat {\tilde \sigma}^2 &= \int_{[0, W]} \int_{[0, W]} \hat {Cov} \left[ \chi \left(N^{\infty} \right)(u), \chi \left(N^{\infty} \right)(v) \right] \hat {\tilde m}(du) \hat {\tilde m}(dv) \\
 &= \int_{[0, W]} \int_{[0, W]} \hat {Cov} \left[ \chi \left(N^{\infty} \right)(u), \chi \left(N^{\infty} \right)(v) \right]r(u)r(v) \hat m(du) \hat m(dv).
 \end{align*}

 Now, recall the parameters $c_1$ and $c_2$ given in \eqref{candc2prorata}. To compute $c_1$ and $c_2$, we first have to estimate $\{ f_1(x) : x = -W, \cdots, T \}$ and $\{ f_2(x) : x = -W, \cdots, T \}$,  where the functions $f_1(\cdot)$ and $f_2(\cdot)$ are defined in \eqref{define_f1} and \eqref{define_f2}. We estimate $\{ \hat f_1(x) : x = -W, \cdots, T \}$ and $\{ \hat f_2(x) : x = -W, \cdots, T \}$ as 
\begin{equation*}
\hat f_1(x) = \left \{\begin{array}{ll}  \int_{[0, T-x] } r(y) \hat m(dy), & \hbox{if $ 0 \le x \le T$},\\
 \int_{[-x, T-x] }  r(y) \hat m(dy), & \hbox{if $ T - W < x < 0$},\\
 \int_{[-x, W] }  r(y) \hat m(dy), & \hbox{if $ -W \le x \le T -W,$} \end{array} \right.\\
  \end{equation*}
  and 
\begin{equation*}
 \hat f_2(x) = \left \{\begin{array}{ll} \frac{1}{n} \sum_{j=1}^n {\left[\int_{[0, T-x]} r(y)M^n_j(dy)\right]}^2 - {\left[ \hat f_1(x)\right]}^2 & \hbox{if $ 0 \le x \le T$},\\
  \frac{1}{n} \sum_{j=1}^n {\left[ \int_{[-x, T-x]} r(y)M^n_j(dy)\right]}^2 - {\left[ \hat f_1(x)\right]}^2, & \hbox{if $ T - W < x < 0$},\\
  \frac{1}{n} \sum_{j=1}^n {\left[\int_{[-x, W]} r(y)M^n_j(dy)\right]}^2 - {\left[ \hat f_1(x)\right]}^2, & \hbox{if $ -W \le x \le T -W,$} \end{array} \right.\\
\end{equation*}
where $r(\cdot)$ is the rebate function given in \eqref{rebatedefn}, $n$ is the total number of items sold and $M^n_j(\cdot)$ is the times of claims measure associated with the $j$-th item sold. Note that if $r \equiv 1$, the estimate of $\{ \hat f_1(x) : x = -W, \cdots, T \}$ and $\{ \hat f_2(x) : x = -W, \cdots, T \}$ is the same as in \eqref{f1estimation} and \eqref{f2estimation}. Now, we integrate by the trapezoid method to estimate $\hat c_1 = \int_{-W}^T \hat f_1(x) \hat \nu(dx)$ and $\hat c_2 = \int_{-W}^T \hat f_2(x) \hat \nu(dx)$, where $\hat \nu(\cdot)$ is an estimate of $\nu(\cdot)$ obtained from the analysis of the sales process. The definition of $\nu(\cdot)$ is given in \eqref{eqn:salesprocess}.

\section{Example}\label{cardata}

We applied our methods to automobile sales and warranty claims data from a large car manufacturer for a single car model and model year. The company warranted each car sold for three years; i.e. $W = 1096$ days. The period for which we are estimating the cost is taken to be a quarter; i.e. $T = 91$ days. This data is the same as the one used by \citet{kulkarni:resnick:2007}, but we do not assume the sales process or the times of claims measure $M(\cdot)$ is Poisson.

Which version of Theorem \ref{mainthm} should we use? To answer this, we analyze the data on claim sizes.

\subsection{Analysis of the claim size distribution}\label{sec:examplesize}

The data consists of the vehicle id which identifies the car, the date on which a car comes with some claim,  the claim id which is unique for each (car, claim) pair and the amount of such a claim. 

From the data, a car on a particular day could come with multiple claims. However, from our definition in \eqref{defn_mjn} of the times of claims measure $M_j^n(\cdot)$ associated with the $j$-th item sold, $M_j^n(\cdot)$ is a point measure consisting of random points $\{C^n_{j,i}, i =1, 2, \cdots\}$. So, to be consistent with our modeling, for each pair (vehicle id, date), we add the costs of all the claims associated with it, i.e. if a car with vehicle id $V$ comes with $p$ claims on a particular date $D$, which cost $X_1, X_2, \cdots, X_p$ respectively, then we assume that the car $V$ has arrived on date $D$ with a single claim of size $X_1 + X_2 + \cdots +X_p$. Thus, we associate the claim size $(X_1 + X_2 + \cdots +X_p)$ to the (vehicle id, date) pair $(V, D)$. Our processing of the data on claim sizes differs from that of \cite{kulkarni:resnick:2007}.

We tabulate the estimated mean, variance and quartiles of the claim size distribution in Table \ref{table:claimsizesummaryfulldata}. Since the estimated third quartile is smaller than the mean, we expect power-like tails of the distribution of claim sizes. The density plot of the claim size distribution and the QQ plot \citep[page 97]{resnickbook:2007} are shown in Figure \ref{densityandQQplot}. We use the QQ estimator \citep[page 97]{resnickbook:2007} to obtain an estimate of $\hat \alpha = 1.52.$ This suggests using version 
\eqref{stablemeanexists} of Theorem \ref{mainthm}.

\begin{table}
\caption{Summary statistics for claim size data.}
\begin{center}
\begin{tabular}{|c|c|c|c|c|c|}
\hline Mean & Variance & First quartile &
Median &  Third quartile \\
\hline 47.53 & 18273.14 & 7.50 & 15.91 & 
42.79\\
\hline
\end{tabular}
\end{center}
\label{table:claimsizesummaryfulldata}
\end{table}

\begin{figure}
\begin{tabular}{cc}
\includegraphics[scale=0.4]{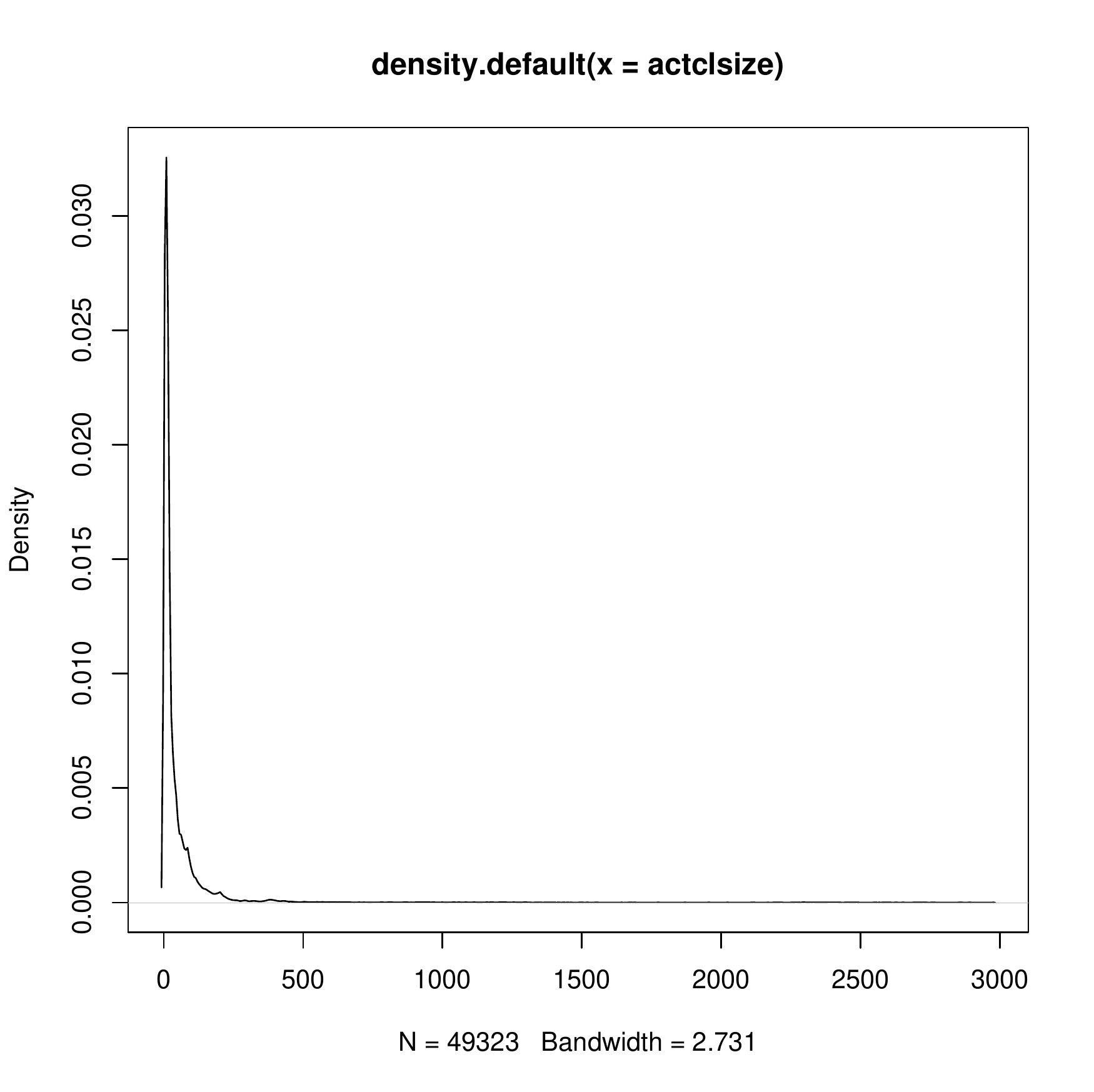}& \includegraphics[scale= 0.4]{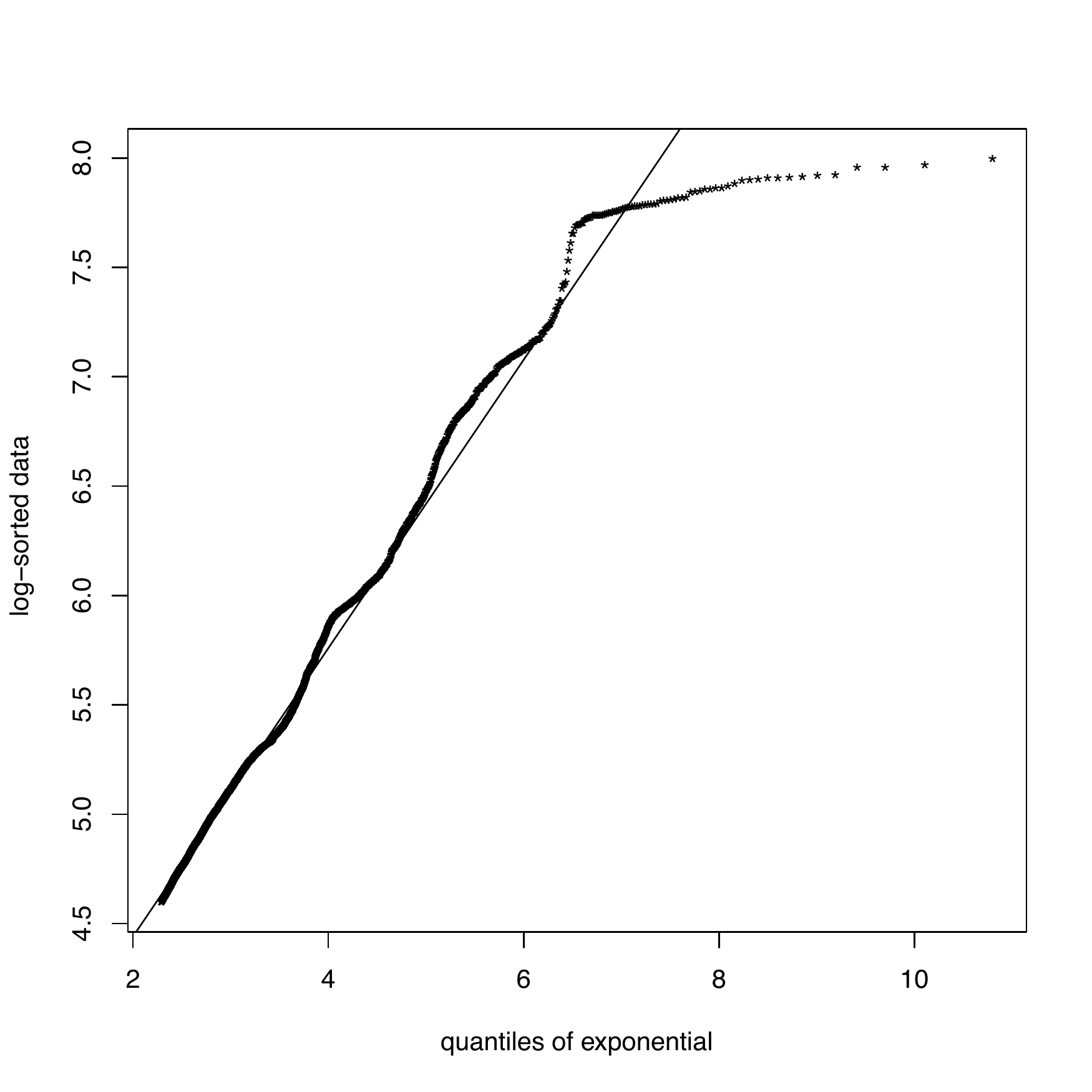} 
\end{tabular}
\caption{ Density and QQ plot (k = 5000) of the claim size distribution (total data points = 49323).}
\label{densityandQQplot}
\end{figure}

However, the density plot shown in Figure \ref{densityandQQplot} almost vanishes after the threshold 500, which suggests that there are few data points which are relatively very large compared to the rest and they are heavily influencing the estimate of $\alpha$. However, there are 459 data points which are bigger than 500. So, on one hand, we cannot discard the claim sizes which are bigger than $500$ as outliers, while on the other hand, a very small proportion of the data (459/49323 = 0.0093) is influencing the summary statistic, the QQ plot and the QQ estimate of $\alpha$. 

We redo the analysis for all the claim sizes which are less than 500. For this case, the summary statistics are tabulated in Table \ref{table:claimsizesummary500} and the density plot and the QQ plot \citep[page 97]{resnickbook:2007} are shown in Figure \ref{densityandQQplot500}. Although the data still seems to have a power-like tail, our estimate of $\alpha$ using the QQ estimator in this case is $\hat \alpha = 2.44$, which, to our dismay, suggests using version \eqref{normalthm} of Theorem \ref{mainthm}.

\begin{table}
\caption{Summary statistics for claim size data of size less than 500.}
\begin{center}
\begin{tabular}{|c|c|c|c|c|c|}
\hline Mean & Variance & First quartile &
Median &  Third quartile \\
\hline 37.38 & 3464.91 & 7.41 & 15.73 &
41.41\\
\hline
\end{tabular}
\end{center}
\label{table:claimsizesummary500}
\end{table}

 \begin{figure}
\begin{tabular}{cc}
\includegraphics[scale=0.4]{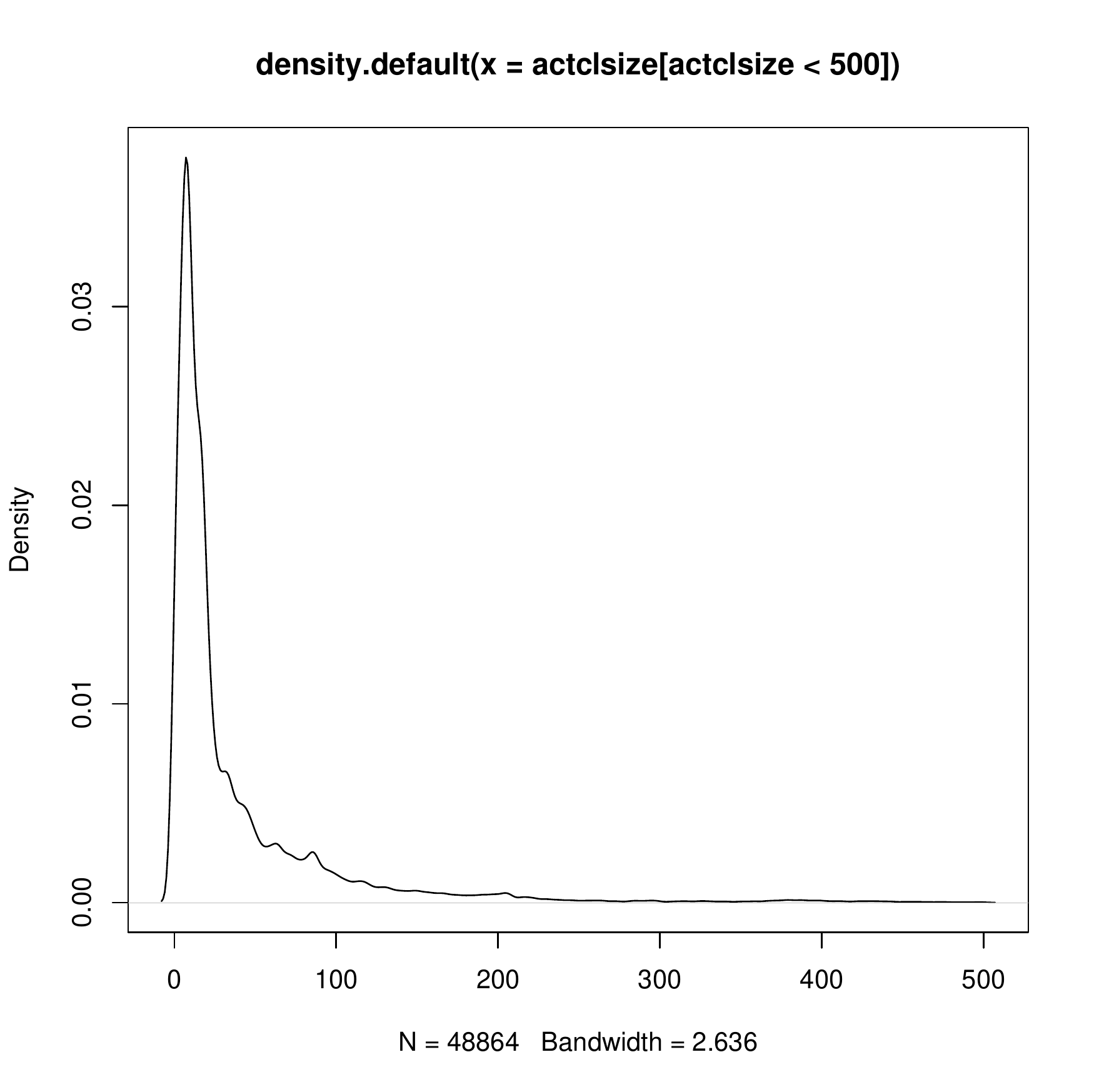}& \includegraphics[scale= 0.4]{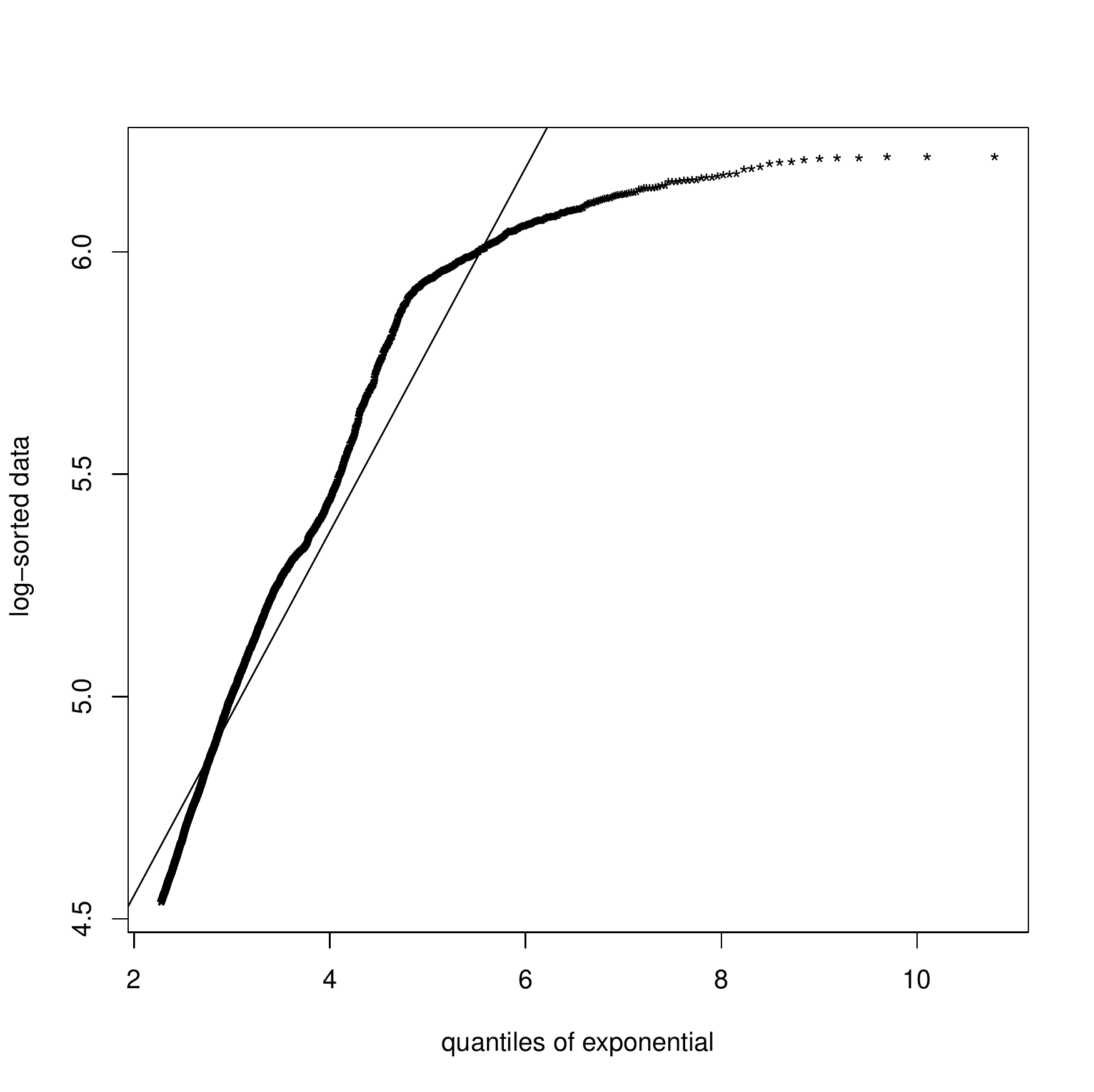} 
\end{tabular}
\caption{ Density and QQ plot (k = 5000) of the distribution of claim sizes which are less than 500(total data points = 48864).}
\label{densityandQQplot500}
\end{figure}

For comparison, we use both the versions \eqref{normalthm}  and \eqref{stablemeanexists} of Theorem \ref{mainthm} and compare the quantiles of the total warranty cost obtained from the two approximations to check robustness of our asymptotic approximation against model error.

\subsection{Analysis of the distribution of the times of claims measure $M(\cdot)$}\label{sec:examplearrivalclaim}
The sales data is needed to compute the times of claims measures $\{ M_j^n(\cdot): j = 1, 2, \cdots, n \}$. The sales data consists of the vehicle id which identifies the car and the date on which it was sold and is a record of 34807 cars sold over a period of 1116 days.

To analyze the claims data, we apply the technique in Section \ref{sec:estimationnormalthm}. To make any estimation about the distribution of the times of claims measure $M(\cdot)$, we obtain the data on $\{ M^n_j((i-1,i]: i = 0, 1, \cdots, W \}$, $j =1, 2, \cdots, n$. Recall, $W =1096$ days. For each vehicle id $j$,  note its date of sale $S^n_j$ and the dates on which it comes with a claim, say it comes with claims on $p$ dates $D_1 < D_2 < \cdots < D_p$. Now, we compute $C^n_{i,j} = D_i - S^n_j, i = 1, 2, \cdots, p.$ Then, we construct the measure $M_j^n(\cdot)$ as $\sum_{i=1}^p \epsilon_{C^n_{i,j}}(\cdot)$. In some cases, we found that $C^n_{i,j} < 0$ (claim honored before the car is sold), or $C^n_{i,j} > W$ (claim honored after the warranty period). We handled this as follows: if  $C^n_{i,j} < 0$, we make it $C^n_{i,j} = 0$ and if $C^n_{i,j} > W$, we make it $C^n_{i,j} = W.$ Thus, we obtain $\{ M^n_j((i-1,i]) := 1_{\{C^n_{i,j} = i \}}: i = 0, 1, \cdots, W \}$, $j =1, 2, \cdots, n$.

We estimate the expected times of claims measure $m(\cdot) = E[M(\cdot)]$ in a manner similar to \cite{kulkarni:resnick:2007}. The plots of  $\{ \left(i, \hat m_1( (i-1, i]) := \frac{1}{n} \sum_{j=1}^n M^n_j((i-1, i]) \right): i = 0, 1, \cdots, W =1096 \}$ and  $\{ \left( i, \hat m_1( (i-1, i] \right): i = 1, \cdots, W-1= 1095 \}$ are shown in Figure \ref{plot:arrivaltimesofclaims}, where $n = 34807$ is the total number of cars in our sales data. Clearly, the plot of $\{ \left( i, \hat m_1( (i-1, i] \right): i = 0, 1, \cdots, 1096\}$ indicates that the measure $m(\cdot)$ has two atoms at $0$ and $W = 1096$. So, we plot $\{\left( i, \hat m_1( (i-1, i] \right): i = 1, \cdots, 1095\}$ to infer the structure of the mean times of claims measure $m(\cdot)$ in the interval $(0, W)$. The linear appearance of $\{ \left( i, \hat m_1( (i-1, i] \right): i = 1, \cdots, 1095\}$ as shown in Figure \ref{plot:arrivaltimesofclaims} suggests that
for $0 < x < 1096$,
\begin{align*}
m(dx) = (ax + b)dx.
\end{align*}
By integrating, we get for $1 \le i \le 1095$, 
\begin{align*}
m( (i -1, i]) = a i + b - \frac{a}{2}.
\end{align*}
From our fitted line over $\{\left( i, \hat m_1( (i-1, i] \right): i = 1, \cdots, 1095\}$ as shown in Figure \ref{plot:arrivaltimesofclaims}, we obtain the estimates 
$$ \hat a = - 0.8872 \times 10^{-6}, \hskip 1 cm \hat b - \frac{\hat a}{2} = 0.1479 \times 10^{-2}.$$
We estimate $m(\{ 0 \})$ and $m(\{ W \})$ by
$$  \hat m_1 (\{ 0 \}) := \frac{1}{n} \sum_{j=1}^n M^n_j(\{0\}) = 0.1330, \hskip 1 cm \hat m_1 (\{ W \}) := \frac{1}{n} \sum_{j=1}^n M^n_j((W-1, W])= 0.0420. $$
Thus, we estimate the measure $m(\cdot)$ as $\hat m(dx) = (\hat ax + \hat b)dx + \hat m_1 (\{ 0 \})1_{\{x = 0\}} +  \hat m_1 (\{ W \})1_{\{x = W\}}$.

 \begin{figure}
\begin{tabular}{cc}
\includegraphics[scale=0.5]{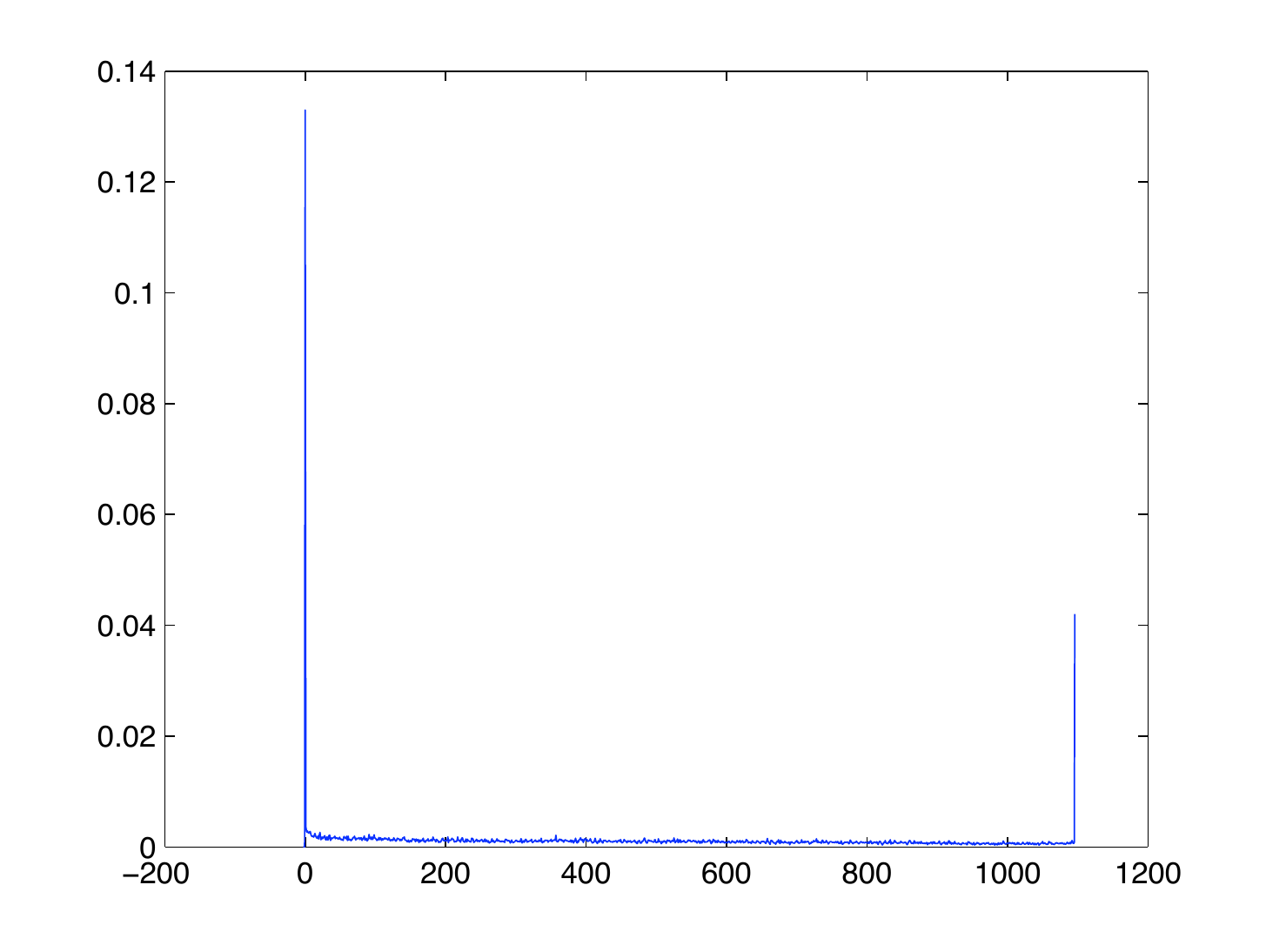}& \includegraphics[scale= 0.5]{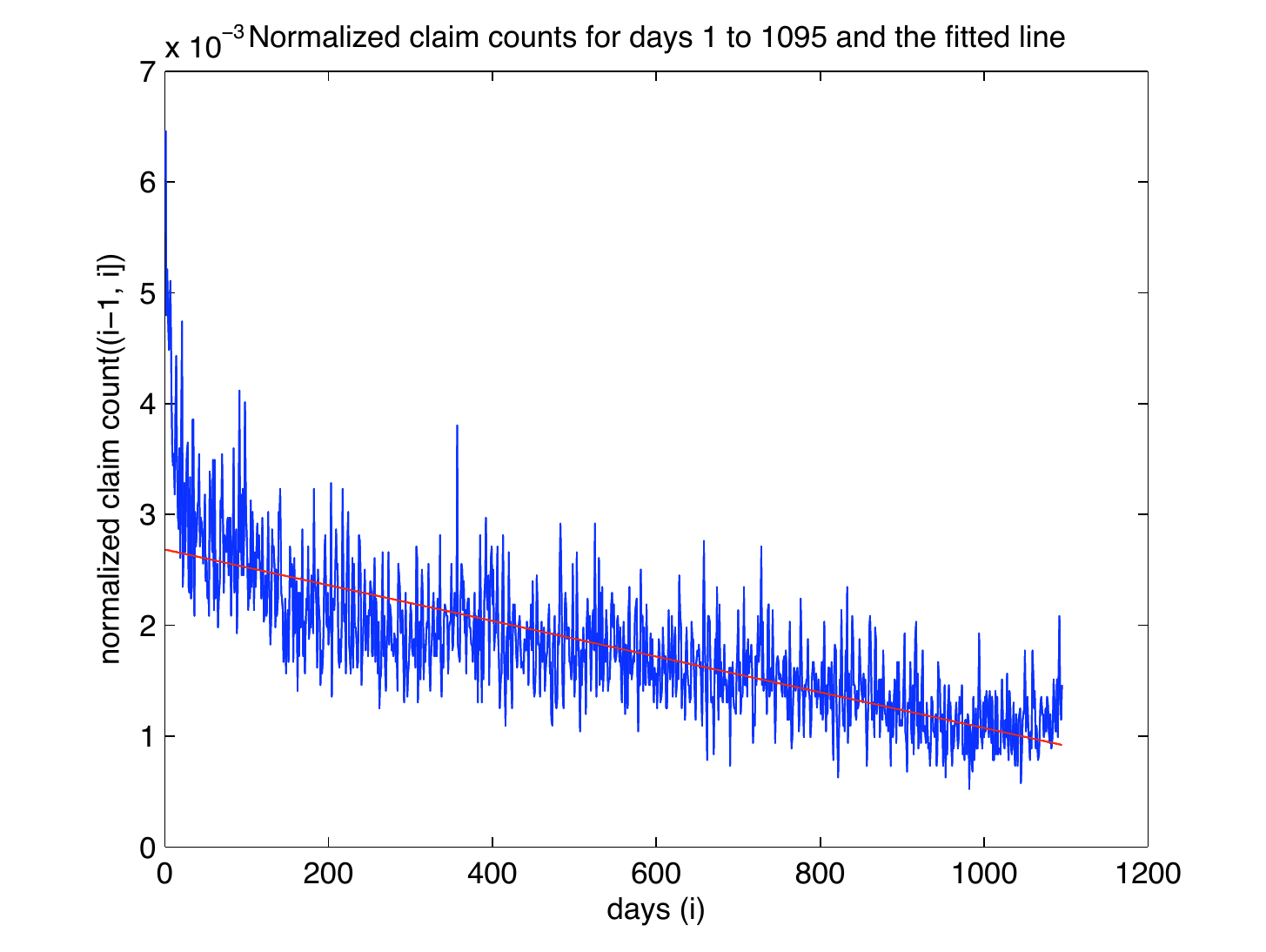} 
\end{tabular}
\caption{Plot of $\{ \left( i, \hat m_1( (i-1, i] \right): i = 0, 1, \cdots, 1096. \}$ and $\{ \left( i, \hat m_1( (i-1, i] \right): i = 1, \cdots, 1095.\}$ .}
\label{plot:arrivaltimesofclaims}
\end{figure}

 To estimate the parameters in the limit distribution of Theorem \ref{mainthm}, version \eqref{normalthm}, we use the estimators of $f_1(\cdot)$ and $f_2(\cdot)$ suggested in \eqref{f1estimation} and \eqref{f2estimation}.

\subsection{Analysis of the distribution of the sales process}\label{sec:examplesales}

We apply the technique explained in Section \ref{sec:estimationnormalthm}. We assume that $n \nu(\cdot)$ follows the Bass model \citep{bass:1969} for the sales period of 1116 days. We choose $n$ as the total sales in those 1116 days and so, $n = 34807$. We use the least squares method discussed in Section \ref{sec:estimationnormalthm} to obtain estimates $(\hat B, \hat C) = ( 4.0149 \times 10^{-4}, 1.6738 \times 10^{-2})$. The time plot of daily count of sales with fitted Bass is given in Figure \ref{plot:salesfittedbass}. 

The fit of Bass model is even better for 12-day counts of sale as shown in Figure \ref{plot:salesfittedbass}. In case of 12-day counts, we obtain the least square estimates $(\hat B, \hat C) = ( 4.0279 \times 10^{-4}, 1.6740 \times 10^{-2})$, which are not too different from the estimates obtained from daily counts. This gives us confidence in our estimates $(\hat B, \hat C) = ( 4.0149 \times 10^{-4}, 1.6738 \times 10^{-2})$ obtained from daily counts and we use these estimates for the following estimation procedure.

 \begin{figure}
\begin{tabular}{cc}
\includegraphics[scale=0.5]{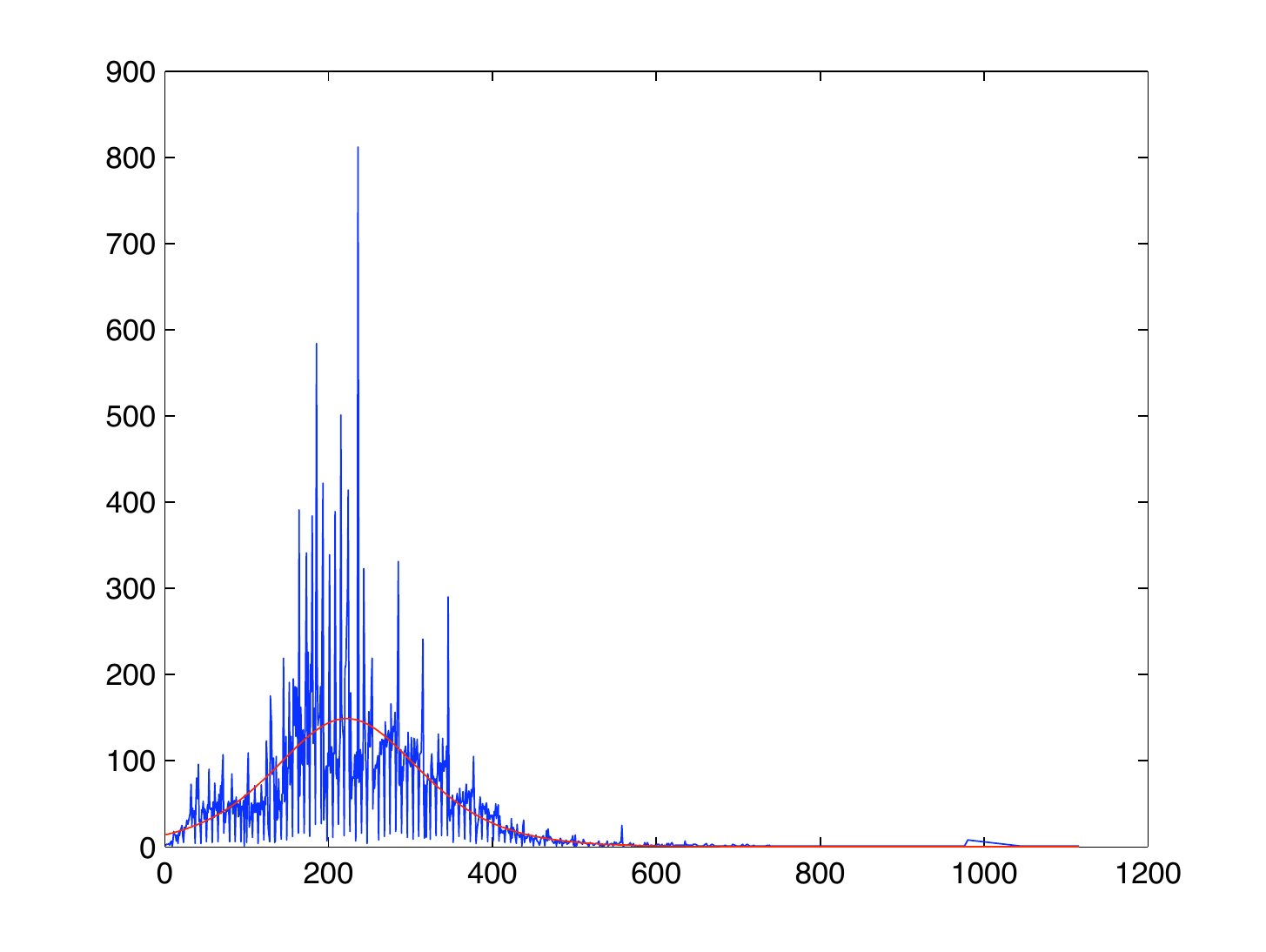}& \includegraphics[scale= 0.5]{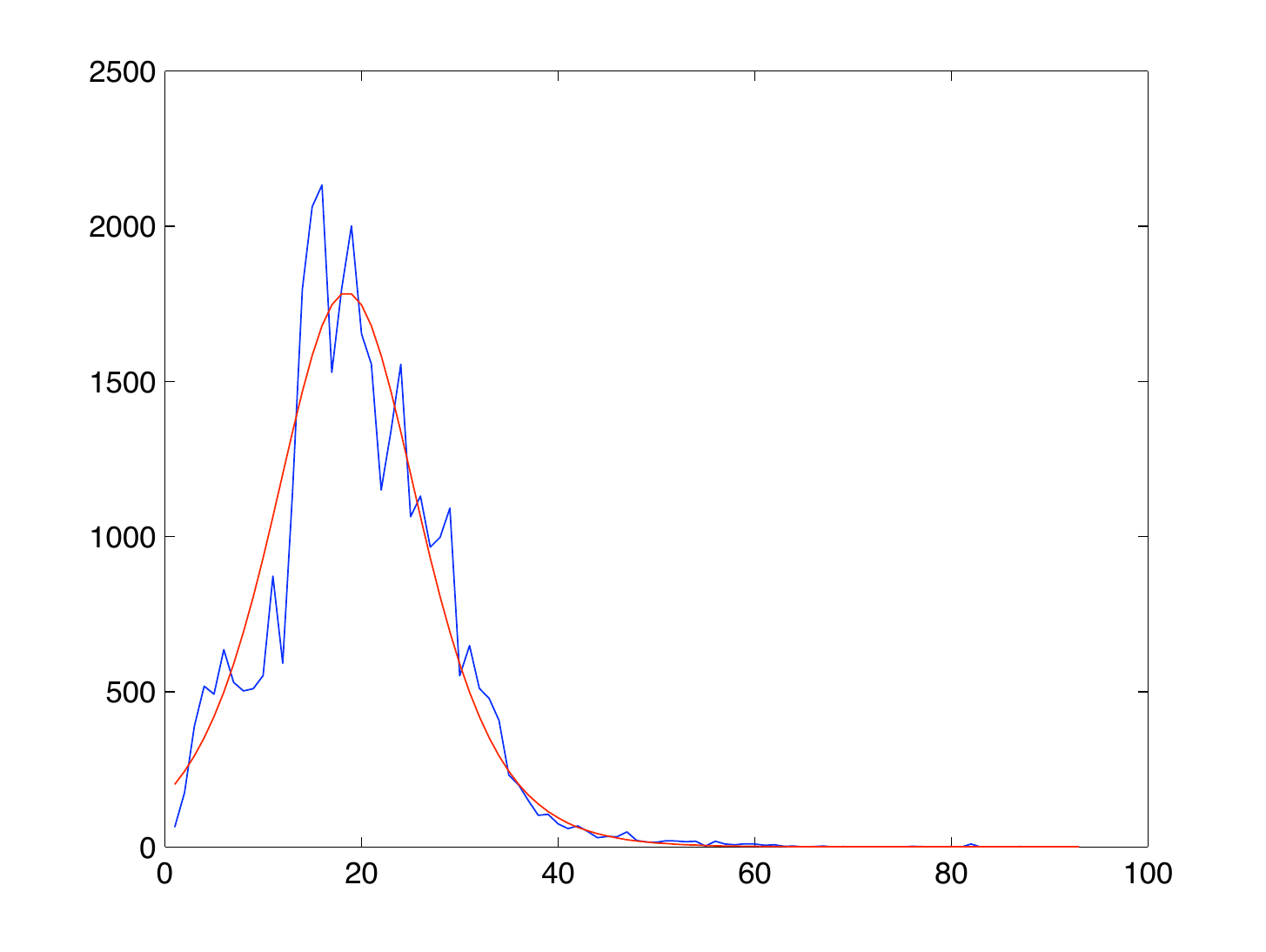} 
\end{tabular}
\caption{Left: Daily counts of sale with fitted Bass, right: 12-day counts of sale with fitted Bass.}
\label{plot:salesfittedbass}
\end{figure}

Recall from Section \ref{sec:estimationnormalthm}, that we have assumed that $\{ Z(t) \}$ is a stationary Gaussian process and the centered and scaled residuals $\{ j_t \}$ will act as surrogates of $\{ Z(t) \}$. We show the time plot and the normal QQ plot of $\{ j_t \}$ in Figure \ref{plot:salesfittedresidual}.

\begin{figure}
\begin{tabular}{cc}
\includegraphics[scale=0.5]{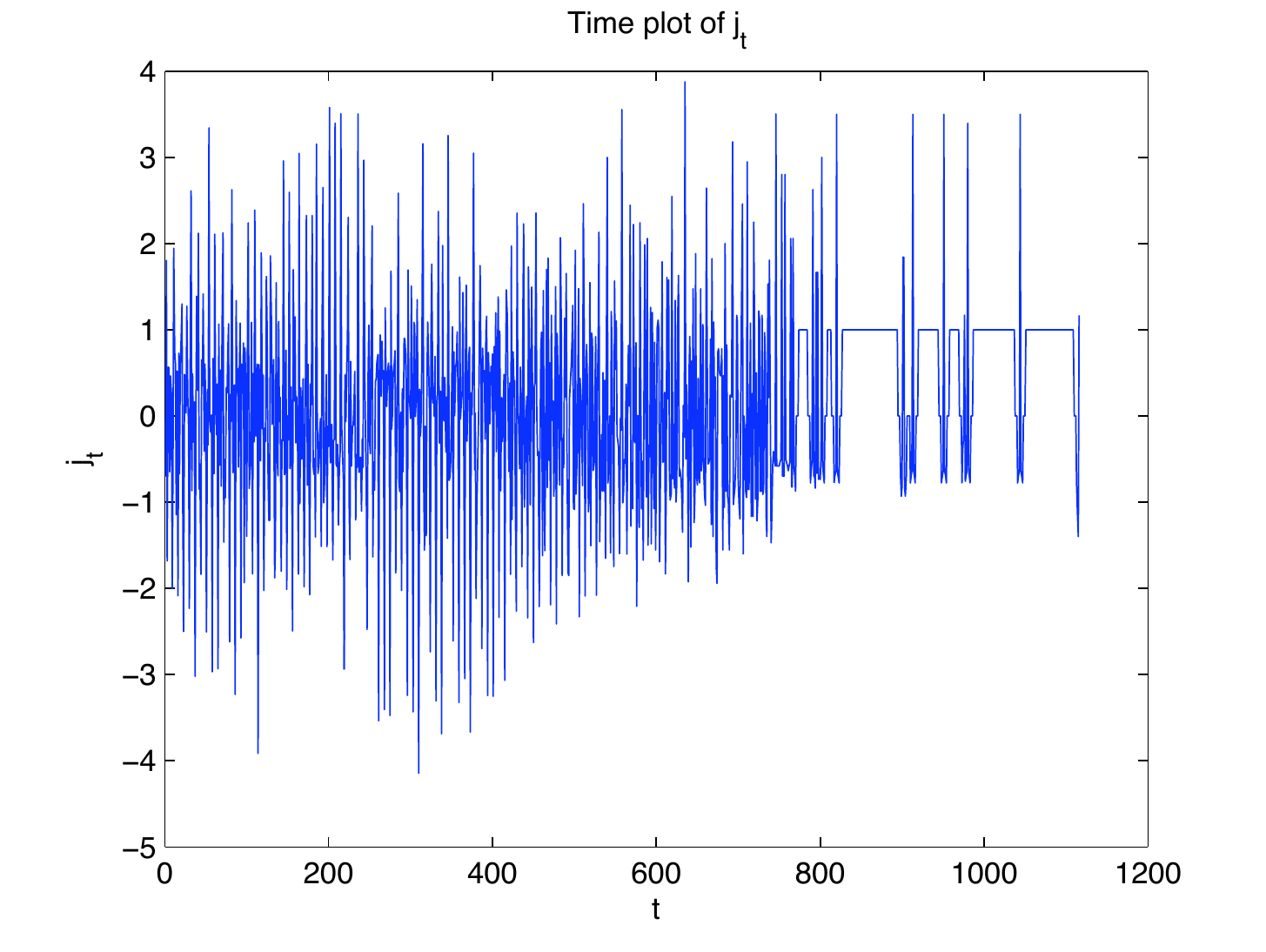}& \includegraphics[scale= 0.5]{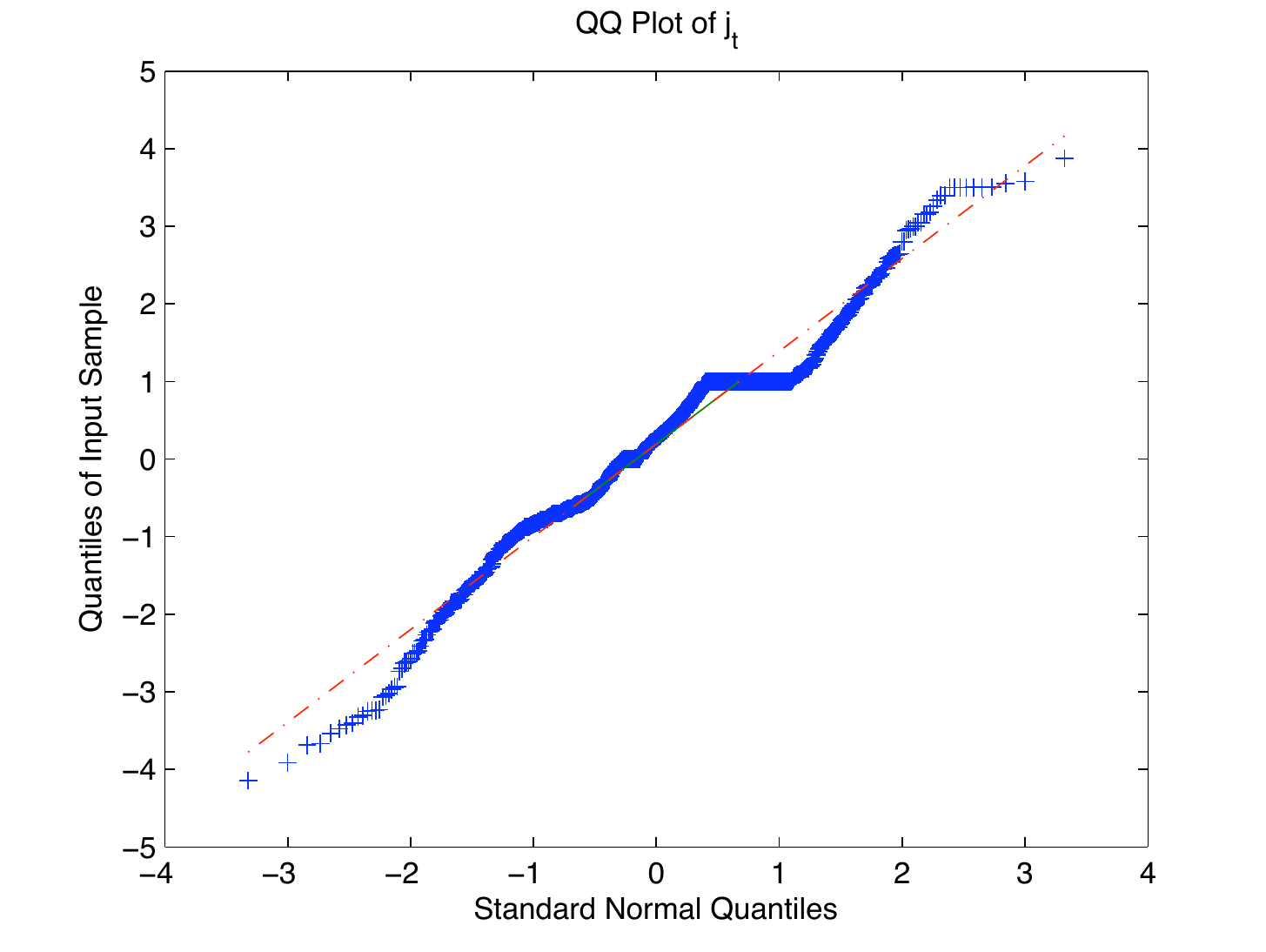} 
\end{tabular}
\caption{Left: Time plot of $\{j_t \}$, right: QQ plot of $\{j_t \}$.}
\label{plot:salesfittedresidual}
\end{figure}

Now, we follow the procedure described in Section \ref{sec:estimationnormalthm} to estimate $\{ \hat \theta (t), \hat \gamma( t, s) \}$.

\subsection{Estimation of quantiles} First, we have to decide the time period for which we  want to approximate the distribution of total warranty cost. We choose two consecutive periods of length 91 days starting from the last sales in date, i.e. we choose the next two quarters from the last sales date. We denote these two quarters as $[0, T]$ and $[T, 2T]$, and accordingly adjust our clock. We assume that $n$ remains the same for the entire period $[-1116, 2 \times 91]$ (recall that we have sales data for a period of $1116$ days).

We compute quantiles of total warranty cost using both the stable and normal approximations and compare them.

For approximation using version \eqref{normalthm} (normal) of Theorem \ref{mainthm}, we follow the method described in Section \ref{sec:estimationnormalthm} to obtain estimates of the six parameters: $c_1, c_2, \tilde \mu$, $\tilde \sigma^2$, $E$ and $V$. However, note that the parameters $c_1, c_2, \tilde \mu$ and $\tilde \sigma^2$ depend on the time-period we are considering, that is the estimates will be different for time-periods $[0, T]$ and $[T, 2T]$. Table \ref{table:estparameters} gives estimates of these parameters for the time periods $[0, T]$ and $[T, 2T]$. For both the time periods $[0, T]$ and $[T, 2T]$, we estimate the parameters $\hat E = 47.53 $ and $\hat V = 18273.14$ using estimates from Table \ref{table:claimsizesummaryfulldata}.

\begin{table}
\caption{Estimated parameters $c_1, c_2, \tilde \mu$ and $\tilde \sigma^2$.}
\begin{center}
\begin{tabular}{|c|c|c|c|c|}
\hline Time-period & $\hat c_1$ & $ \hat c_2$ & $\hat {\tilde \mu}$ & $\hat {\tilde \sigma}^2$ \\
\hline $[0, T]$ & 0.0614 & 0.0887 & 1.0210 & 1.5568\\
\hline $[T, 2T]$ & 0.0540 & 0.0818 & 0.8817 & 0.9712\\
\hline
\end{tabular}
\label{table:estparameters}
\end{center}
\end{table}

For approximation using version \eqref{stablemeanexists} (stable) of Theorem \ref{mainthm}, we follow the method of estimation described in Section \ref{sec:estimationstablemeanexists}. We obtain estimates of $c_1$ for the time-periods $[0, T]$ and $[T, 2T]$ from Table \ref{table:estparameters}. We estimate the parameter $E = 47.53 $ using estimates from Table \ref{table:claimsizesummaryfulldata}. We estimate $\hat \alpha = 1.52$ using the QQ estimator (k = 5000)\citep[page 97]{resnickbook:2007} and obtain $\hat b(n) = n^{1/\hat \alpha}$. Using this estimate of $\alpha$ and \eqref{stablemeanpara}, we estimate of the parameters of the distribution of $Z_{\alpha}(1)$ (following the parametrization of \cite{samorodnitsky:taqqu:1994}, as described in Section \ref{sec:estimationstablemeanexists}) as $\hat \mu = 0, \hskip 0.1 cm \hat \sigma = 1.8688$ and $\hat \beta = 1$. We use J. P. Nolan's software available at http://academic2.american.edu/$\sim$jpnolan/stable/stable.html to compute the stable quantiles.

The quantiles of total warranty cost using both the approximations: version \eqref{normalthm} and version \eqref{stablemeanexists} of Theorem \ref{mainthm}, are listed in Table \ref{tablequantiles}. Note that the quantiles of total warranty cost computed using version \eqref{normalthm} of Theorem \ref{mainthm} are bigger than those computed using  version \eqref{stablemeanexists} of Theorem \ref{mainthm}, but the difference is not huge as one might have expected since version \eqref{stablemeanexists} is applicable for the heavy tailed data whereas version \eqref{normalthm} is applicable for the light-tailed data.

\begin{table}
\caption{Quantiles for the total cost on warranty claims.}
\begin{center}
\begin{tabular}{|c|c|c|c|c|}
\hline 
 & \multicolumn{2}{|c|}{Time period $[0, T]$} & \multicolumn{2}{|c|}{ Time period $[T, 2T]$}\\
\hline
 $p$ & $p$-th quantile & $p$-th quantile  & $p$-th quantile  & $p$-th quantile  \\
 & using version \eqref{normalthm}  & using version \eqref{stablemeanexists}  & using version \eqref{normalthm}  & using  version \eqref{stablemeanexists} \\
& of Theorem \ref{mainthm} & of Theorem \ref{mainthm} & of Theorem \ref{mainthm} & of Theorem \ref{mainthm}\\
\hline
0.50 & 110,694.91  & 101,448.27& 97,219.87 & 89,224.58\\
\hline
0.75 & 119,449.01  & 101,791.20 &  104,532.99 & 89,539.76\\
\hline
0.80 & 121,618.18  & 101,897.93 &  106,345.11 & 89,637.85\\
\hline
0.85 & 124,146.62  & 102,040.29 & 108,457.35 & 89,768.68\\
\hline
0.90 & 127,327.97  & 102,258.94 & 111,115.03 & 89,969.64\\
\hline
0.95 & 132,043.22  & 102,723.28 & 115,054.12 & 90,396.39\\
\hline
0.99 & 140,888.23  &  104,857.40 & 122,443.19 & 92,357.76\\
\hline
\end{tabular}
\label{tablequantiles}
\end{center}
\end{table}

We computed the actual number of claims and the total warranty cost for the periods $[0, T]$ and $[T, 2T]$ from our data. Though we cannot test the fit of a distribution from a single observation, we do some sanity checks to decide how well the approximations work. 

Start with the actual number of claims. The number of claims in $[0, T]$ is $R^n_{[0, T]} = 2352$ and the number of claims in $[T, 2T]$ is $R^n_{[T, 2T]} = 1516.$ Let $AF_{R^n_P}(\cdot)$ be the approximation of the distribution function of the total number of claims that arrived in period $P$ using Theorem \ref{numberclaims}. We compute $AF_{R^n_P}(R^n_P)$ for both the time periods $P = [0, T]$ and $P = [T, 2T]$. If $AF_{R^n_P}(\cdot)$ were the actual distribution function of $R^n_P$, then $AF_{R^n_P}(R^n_P)$ would be uniform on $[0, 1]$. Our computed $AF_{R^n_P}(R^n_P)$ values are $AF_{R^n_{[0, T]}}(R^n_{[0, T]}) = 0.5381$ and $AF_{R^n_{[T, 2T]}}(R^n_{[T, 2T]}) = 0.0029$. The probability that a $Uniform([0, 1])$ random variable take a value more extreme than a number $a$ is $2\min \left\{ P[U \le a], P[U > a]\right\}$, where $U \sim Uniform([0, 1])$, which, for the numbers $\{ AF_{R^n_P}(R^n_P): P = [0, T], [T, 2T] \}$ are 0.9238 and 0.0058 respectively. These values suggest that $\{ AF_{R^n_P}(\cdot): P = [0, T], [T, 2T] \}$ may be reasonable fits for the distributions of $\{R^n_P: P = [0, T], [T, 2T] \}$.

Now, we compute the actual costs for the time periods $[0, T]$ and $[T, 2T]$, denoted by $\cost([0, T])$ and $\cost([T, 2T])$ respectively. Let $A_1F_P(\cdot)$ and $A_2F_P(\cdot)$ be the approximate distribution functions of the total warranty cost using versions \eqref{normalthm} and \eqref{stablemeanexists} of Theorem \ref{mainthm} respectively for the period $P$. The computed values of $\{ A_iF_P(\cost(P)), i = 1, 2, P = [0, T], [T, 2T] \}$ are noted in Table \ref{tablefitting}. For computing $\{ A_2F_P(\cost(P)), P = [0, T], [T, 2T] \}$, we used J. P. Nolan's software available at http://academic2.american.edu/$\sim$jpnolan/stable/stable.html. If $A_iF_P(\cdot)$ is the actual distribution of $\cost(P)$, then $A_iF_P(\cost(P))$ would be uniform on $[0, 1]$. The probabilities that a $Uniform([0, 1])$ random variable take a value more extreme (as explained before in the previous paragraph) than $A_1F_P(\cost(P))$ is greater than that of $A_2F_P(\cost(P))$ for both the periods $[0, T]$ and $[T, 2T]$, which suggest that for our data on sales and warranty claims of cars, the approximation of the distribution of total warranty cost using version \eqref{normalthm} of Theorem \ref{mainthm} is better than the approximation using version \eqref{stablemeanexists} of Theorem \ref{mainthm} for both the periods $[0, T]$ and $[T, 2T]$. Comparing the quantiles of $\{ A_iF_P(\cdot), i = 1, 2, P = [0, T], [T, 2T] \}$ given in Table \ref{tablequantiles} with the actual costs $\{ \cost(P): P = [0, T], [T, 2T] \}$ given in Table \ref{tablefitting}, we arrive at the same conclusion.

\begin{table}
\caption{Computed values of $\{ A_iF_P(\cost(P)), i = 1, 2, P = [0, T], [T, 2T] \}$.}
\begin{center}
\begin{tabular}{|c|c|c|c|}
\hline 
 Time period   & Actual cost in & Approximation using & Approximation using\\
 $P$ & time-period $P$ & version \eqref{normalthm}  of Theorem \ref{mainthm} & version \eqref{stablemeanexists} of Theorem \ref{mainthm}\\
 & $\cost(P)$ & $A_1F_P(\cost(P))$ & $A_2F_P(\cost(P))$\\
 \hline 
 $[0, T]$ & 148,180.60 & 0.9981 & 0.9998\\
 \hline
 $[T, 2T]$ & 98,992.90 & 0.5649 & 0.9983 \\
\hline
\end{tabular}
\label{tablefitting}
\end{center}
\end{table}

\section{Concluding remarks}

We have approximated the distribution of the total warranty claims expenses incurred in a fixed period. Our assumptions on the distribution of the sales process $N^n(\cdot)$ and the times of claims measure $M(\cdot)$ are mild and hence our approximation is applicable in a general context. However, we have introduced a lot of independence in our modeling. For example, we have assumed that the claim sizes are iid, but in practice this may not true. Similarly, the sales process $N^n(\cdot)$ and the times of claims measures $\{ M_j^n(\cdot) \}$ or the sales process $N^n(\cdot)$ and the claim sizes may be dependent. We have ignored such dependences, but allowing for dependence might lead to a more realistic modeling and better approximation.

Our estimation procedure is mostly non-parametric and hence generally applicable. However, we have assumed a parametric form for $\nu(\cdot)$ using the model proposed by \citet{bass:1969}. Estimating $\nu(\cdot)$ non-parametrically might lead to robustness against model error. 

Another issue is the choice of $n$ in our approximation. We interpret $n$ as a measure of sales volume. Though total sales in our observed sales data is a natural candidate for $n$ as we have argued, it is not the only candidate. Since $n$ plays an important role in the approximation, the choice of $n$ might have a significant impact. 

In Section \ref{cardata}, we have demonstrated the applicability of our method. However, a company deciding on reserves to cover warranty cost next quarter should use a more complete and carefully collected dataset.

\section{proofs}\label{sec:proofs}

To prove the asymptotic results, first we state and prove two lemmas.

\begin{lem}\label{prod_approx_lemma}
\rm{Under the Assumptions \ref{freerepl_sale}, \ref{freerepl_claim1} and \ref{freerepl_claim2} of Section \ref{freereplass}, for $\lambda \in \R,$
\begin{align}\label{prod_approx}
 \prod_{ \{j : S_j^n \in [-W, T] \}}  E^{S_j^n} \left[ \exp \left( \right. \right. & \left. \left. i \lambda n^{-1/2} \left(\delta( S_j^n) - f_1(S^n_j) \right) \right)   \right] \nonumber \\
& \qquad = \prod_{ \{j : S_j^n \in [-W, T] \}}  E^{S_j^n} \left[ 1  - \lambda^2 {(2n)}^{-1} {\left(\delta(S^n_j) - f_1(S_j^n)\right)}^2  \right]  + o_p(1).
\end{align}
}
\end{lem}

\begin{proof}[Proof of Lemma \ref{prod_approx_lemma}]
Using the fact that for all $n$, $|\prod_{i=1}^n a_i - \prod_{i=1}^n b_i | \le \sum_{i=1}^n |a_i - b_i|,$ for $|a_i|, |b_i| \le 1$, we get
\begin{align}\label{prod_approx_proof}
 \Big{|} & \prod_{ \{j : S_j^n \in [-W, T] \}}E^{S_j^n} \left[ \exp \left( i \lambda n^{-1/2}  \left(\delta( S_j^n) - f_1(S^n_j) \right) \right)  \right]  \nonumber\\
& \hskip 5 cm  - \prod_{ \{j : S_j^n \in [-W, T] \}}  E^{S_j^n} \left[ 1  - \lambda^2 {(2n)}^{-1} {\left(\delta(S^n_j) - f_1(S_j^n)\right)}^2  \right] \Big{|} \nonumber\\
 &\le \sum_{ \{j : S_j^n \in [-W, T] \}} \Big{|} E^{S_j^n} \left[ \exp \left( i \lambda n^{-1/2} \left(\delta( S_j^n) - f_1(S^n_j) \right) \right)  - 1 +  \lambda^2{(2n)}^{-1} {\left(\delta(S^n_j) - f_1(S_j^n)\right)}^2  \right]  \Big{|} \nonumber \\
 &= \sum_{ \{j : S_j^n \in [-W, T] \}} \Big{|} E^{S_j^n} \left[ \exp \left( i \lambda n^{-1/2} \left(\delta( S_j^n) - f_1(S^n_j) \right) \right)  - 1 - i \lambda n^{-1/2} \left(\delta( S_j^n) - f_1(S^n_j) \right) \right.\nonumber \\
 & \hskip 8 cm  \left. +  \lambda^2{(2n)}^{-1} {\left(\delta(S^n_j) - f_1(S_j^n)\right)}^2  \right]  \Big{|} \nonumber \\
&  \le \sum_{ \{j : S_j^n \in [-W, T] \}} E^{S_j^n} \left[ \left( \frac{\lambda^3}{6n\sqrt{n}} {|\delta(S^n_j) - f_1(S_j^n)|}^3\right) \wedge \left( \frac{\lambda^2}{n}{\left(\delta(S^n_j) - f_1(S_j^n)\right)}^2 \right) \right] .\nonumber\\
\intertext{Since for all $x$, $\delta(x)$ and $f_1(x)$ are bounded by $M([0,W])$ and $m([0, W])$ respectively, the above quantity is bounded by}
& \sum_{ \{j : S_j^n \in [-W, T] \}} E \left[ \left[ \frac{\lambda^3}{6n\sqrt{n}} {\left(M([0, W]) + m([0, W])\right)}^3 \right] \wedge  \left[ \frac{\lambda^2}{n} {\left(M([0, W]) + m([0, W])\right)}^2 \right] \right] \nonumber \\
&   \le \frac{\lambda^2N^n([-W, T])}{n}E \left[ \left[ \frac{\lambda}{6\sqrt{n}}{\left(M([0, W]) + m([0, W])\right)}^3\right] \wedge \left[ {\left(M([0, W]) + m([0, W])\right)}^2 \right] \right]  \nonumber \\
& \stackrel{P}{\rightarrow} 0.
\end{align}

The convergence in the last step holds by noting that $N^n([-W, T])/n \stackrel{P}{\rightarrow} \nu([-W, T])$ by Assumption \ref{freerepl_sale} of Section \ref{freereplass} and the quantity within the expectation converges to 0 using the dominated convergence theorem. To understand how we use the dominated convergence theorem, first note that the quantity inside the expectation is dominated by ${\left(M([0, W]) + m([0, W])\right)}^2,$ which has a finite expectation by Assumption \ref{freerepl_claim2} of Section \ref{freereplass}. On the other hand, the quantity inside the expectation is also dominated by  $\frac{\lambda}{6\sqrt{n}}{\left(M([0, W]) + m([0, W])\right)}^3,$ which converges to 0 almost surely. Hence, using the dominated convergence theorem, we get as $n \to \infty$,
\begin{equation*}
E \left[ \left[ \frac{\lambda}{6\sqrt{n}}{\left(M([0, W]) + m([0, W])\right)}^3\right] \wedge \left[ {\left(M([0, W]) + m([0, W])\right)}^2 \right] \right] \rightarrow 0.
\end{equation*}
\end{proof}

The sales process $N^n(\cdot)$ is a non-decreasing process on $[-W, T]$ and hence, induces a measure on $[-W, T]$. In the following, we refer to $N^n(\cdot)$ to mean both the sales process in $D([-W, T])$ and the measure it induces. It should be clear from the context what we mean by $N^n(\cdot)$. The same rule of notation holds for the non-decreasing function $\nu(\cdot)$ defined in \eqref{eqn:salesprocess}. Now, we state the second lemma.

\begin{lem}\label{use_sales_conv}
\rm{ Under Assumptions \ref{freerepl_sale}, \ref{freerepl_claim1.5} and \ref{freerepl_claim2} of Section \ref{freereplass}, 
\begin{enumerate}
\item \label{firstpart} The integral of $f_1(\cdot)$ (defined in \eqref{define_f1}) with respect to the centered and scaled sales process converges weakly to a Gaussian random variable, i.e.
\begin{equation}\label{fubini_lem_1}
 \int_{[-W, T]} f_1(x) \left(\frac{ N^n - n\nu}{\sqrt{n}} \right) (dx) \Rightarrow \int_{[0, W]}   \chi \left(N^{\infty} \right)(u) \tilde m(du),
\end{equation}
where $\nu(\cdot)$ and $N^{\infty}(\cdot)$ are given in \eqref{eqn:salesprocess}, the measure $\tilde m(\cdot)$ defined in \eqref{define_tilde_m}, the function $\chi(\cdot)$ is defined in \eqref{defineg} and the Gaussian random variable $\int_{[0, W]} \chi \left(N^{\infty} \right)(u) \tilde m(du)$ has mean $\tilde \mu$ and variance $\tilde \sigma^2$ given in \eqref{dotmudotsigma}.
\item \label{secondpart} The integral of $f_2(\cdot)$ (defined in \eqref{define_f2}) with respect to the sales process scaled by $n$ converges in probability to a constant, i.e. 
\begin{equation}\label{fubini_lem_2}
 \frac{1}{n} \int_{[-W, T]} f_2(x)N^n(dx) \stackrel{P}{\rightarrow} \int_{[-W, T]}   f_2(x)\nu(dx),
 \end{equation}
 where $\nu(\cdot)$ is given in \eqref{eqn:salesprocess}.
\end{enumerate}
}
\end{lem}

\begin{proof}[Proof of Lemma \ref{use_sales_conv}]

\eqref{firstpart} Step 1: We assume $2T < W$. So, $T < W -T$. First, note that, by using Fubini's theorem, we get the following three equations
\begin{align}\label{fubini1}
 \int_{[0, T]} \tilde m([0, T-x]) \left(\frac{ N^n - n \nu }{\sqrt{n}} \right) (dx) = \int_{[0, T]} \left(\frac{N^n - n \nu}{\sqrt{n}} \right) ([0, T - u]) \tilde m(du),
\end{align} 
\begin{align}\label{fubini2}
 \int_{(T-W, 0)} \tilde m([-x, T-x]) \left(\frac{N^n - n \nu}{\sqrt{n}} \right) (dx) &= \int_{(0, T]} \left(\frac{ N^n - n \nu}{\sqrt{n}} \right) ([-u, 0)) \tilde m(du) \\
 & \qquad + \int_{(T, W-T)} \left(\frac{ N^n - n \nu}{\sqrt{n}} \right) ([-u, T - u]) \tilde m(du) \nonumber \\
 & \qquad + \int_{[W-T, W)} \left(\frac{ N^n - n \nu}{\sqrt{n}} \right) ((T-W, T - u]) \tilde m(du), \nonumber
\end{align} 
and
\begin{align}\label{fubini3}
 \int_{[-W, T-W]} \tilde m([-x, W]) \left(\frac{N^n - n \nu}{\sqrt{n}} \right) (dx) = \int_{[W-T, W]} \left(\frac{ N^n - n \nu}{\sqrt{n}} \right) ([-u, T - W]) \tilde m(du).
\end{align}

Using \eqref{fubini1}, \eqref{fubini2}, \eqref{fubini3} and the definition of $\chi$ from \eqref{defineg}, we get,

\begin{align}\label{fubini4}
 \int_{[-W, T]} f_1(x) \left(\frac{ N^n - n \nu }{\sqrt{n}} \right) (dx) &= \int_{[0, W]}   \chi \left(\frac{ N^n - n \nu}{\sqrt{n}} \right)(u) \tilde m(du).
 \end{align} 
 
Step 2: We use the continuous mapping theorem to prove our result. First, define $\xi: D([-W, T]) \mapsto \R$ by
\begin{equation}\label{defineh}
\xi(x) = \int_{[0, W]} \chi(x)(u) \tilde m(du), \hskip 1 cm x \in D([-W, T]),
\end{equation}
where $\chi(\cdot)$ is defined in \eqref{defineg}. 

We show that the continuous functions of $D([-W, T])$ are continuity points of the function $\xi : x \mapsto  \int_{[0, W]} \chi(x)(u) \tilde m(du)$. The definition of $\xi$ is similar to the well-known convolution functions \citep[page 143]{feller:1971}. Suppose, $x_n \to x$ in $D([-W, T])$ in the Skorohod topology and $x$ is continuous. Then, $x_n \to x$ uniformly in $[-W, T]$ \cite[page 124]{billingsley:1999} and 
 \begin{align*}
  \Big{|} \int_{[0, W]} \chi(x_n)(u) \tilde m(du) -  \int_{[0, W]} \chi(x)(u) \tilde m(du) \Big{|} &\le  \sup_{0 \le u \le W} |\chi(x_n)(u) - \chi(x)(u) | \tilde m([0, W])\\
  &\le  \sup_{-W \le u \le T} 2|x_n(u) - x(u)| \tilde m([0, W]) \rightarrow 0,
 \end{align*}
 as $ n \to \infty$. So, the discontinuity points of $\xi(\cdot)$ are contained in $D([-W, T]) \backslash C([-W, T])$. 

 Since the limiting process $N^{\infty}(\cdot) \in C([-W, T])$ and $\xi(\cdot)$ is continuous on $C([-W, T])$, by the continuous mapping theorem \cite[page 21]{billingsley:1999},
\begin{equation*}
\xi \left(\frac{ N^n - n \nu}{\sqrt{n}} \right) \Rightarrow \xi (N^{\infty})
\end{equation*}
on $\R,$ i.e., using \eqref{defineh},
\begin{align}\label{hconv}
\int_{[0, W]} \chi \left( \frac{N^n - n \nu }{\sqrt{n}}  \right)(u) \tilde m(du) \Rightarrow \int_{[0, W]} \chi\left(N^{\infty} \right)(u) \tilde m(du).  
\end{align}
Hence, by \eqref{fubini4} and \eqref{hconv}, we have proved part \eqref{firstpart} of Lemma \ref{use_sales_conv}.

\eqref{secondpart} Using Assumption \ref{freerepl_sale} of Section \ref{freereplass}, we get
$N^n(\cdot)/n \Rightarrow \nu(\cdot)$ on $D([-W, T])$. Since Skorohod convergence implies vague convergence \citep{jagers:1972}, we get 
\begin{align}\label{skorohod_to_vague}
\frac{N^n}{n}(\cdot) \Rightarrow \nu(\cdot)
\end{align}
in $M_+([-W, T])$ (see Section \ref{notations}). By Assumption \ref{freerepl_claim1.5} of Section \ref{freereplass} and the definition of $\delta(\cdot)$ given in \eqref{define_tilde_f}, the random function $\delta(\cdot)$ is almost surely continuous at all points except at most at $T-W$ and $0$. Hence, using definition of $f_2(\cdot)$ in \eqref{define_f2}, Assumption \ref{freerepl_claim2} of Section \ref{freereplass} and the dominated convergence theorem, we get that $f_2(x)$ is discontinuous at most at two points, $T-W$ and $0.$ So, $f_2(x)$ can be written as
\begin{equation}\label{f2decomp}
f_2(x) = \left \{ \begin{array}{ll} f_{2, c+}(x) + d_1, & \hbox{if $-W \le x \le T-W,$}\\f_{2, c+}(x) + d_2, & \hbox{if $T-W < x < 0,$}\\f_{2, c+}(x) + d_3, & \hbox{if $0 \le x \le T,$}\end{array} \right.
\end{equation}
where $f_{2, c+}(x)$ is a continuous non-negative function and $d_1, d_2, d_3$ are three constants. Therefore, we get 
\begin{align}\label{f2intdecomp}
 \frac{1}{n} \int_{[-W, T]} f_2(x)N^n(dx) &=  \frac{1}{n} \int_{[-W, T]} f_{2, c+}(x)N^n(dx) + \frac{d_1}{n}N^n([-W, T-W]) \nonumber \\
 &\qquad + \frac{d_2}{n}N^n((T-W, 0)) + \frac{d_3}{n}N^n([0, T]).
\end{align}
By Assumption \ref{freerepl_sale} of Section \ref{freereplass}, $\nu(\cdot)$ is continuous at $T-W$ and $0$.  Hence, using \eqref{skorohod_to_vague} we get that the last three terms on the right side of \eqref{f2intdecomp} converge in probability to $d_1\nu([-W, T-W]), \hskip 0.1 cm d_2\nu((T-W, 0))$ and $d_3\nu([0, T])$ respectively. Also, by \eqref{skorohod_to_vague} we get 
\begin{equation*}
\frac{1}{n} \int_{[-W, T]} f_{2, c+}(x)N^n(dx) \stackrel{P}{\rightarrow}  \int_{[-W, T]} f_{2, c+}(x)\nu(dx).
\end{equation*}
Hence, using \eqref{f2decomp} we get
\begin{align*}
 \frac{1}{n} \int_{[-W, T]} f_2(x)N^n(dx) &=  \frac{1}{n} \int_{[-W, T]} f_{2, c+}(x)N^n(dx) + \frac{d_1}{n}N^n([-W, T-W]) \\
 &\qquad + \frac{d_2}{n}N^n((T-W, 0)) + \frac{d_3}{n}N^n([0, T]) \\
 &\stackrel{P}{\rightarrow} \int_{[-W, T]} f_{2, c+}(x)\nu(dx) + \frac{d_1}{n}\nu([-W, T-W])\\
 &\qquad + \frac{d_2}{n}\nu((T-W, 0)) + \frac{d_3}{n}\nu([0, T]) \\
 &= \int_{[-W, T]} f_2(x)\nu(dx).
\end{align*}
\end{proof}

\begin{proof}[Proof of Theorem \ref{numberclaims}]
Recall that $R^n_j$ denotes the total number of claims in $[0, T]$ for the $j$-th item sold defined in \eqref{define_rnj} and $R^n$ denotes the total number of claims in $[0, T]$ defined in \eqref{define_rn}. Using Assumption \ref{freerepl_claim1} of  \ref{freereplass}, we get
\begin{align} \label{claimgivensales}
R^n_j | S_j^n \stackrel{d}{=} \delta (S_j^n).
\end{align}
Step 1: From Assumption \ref{freerepl_claim1} of Section \ref{freereplass}, we also get that given the sales process $N^n(\cdot)$, the random variables $\{ R^n_j: j \ge 1\}$ are independent. The characteristic function of the centered and scaled $R^n$ given the sales process $N^n(\cdot)$ is for $\lambda \in \R,$
\begin{align*}
E^{N_n(\cdot )}&\left[ \exp \left( i \lambda n^{-1/2}\left(R^n - \int_{[-W,T]} f_1(x) N^n(dx) \right)\right)   \right] \nonumber \\
& \qquad = E^{N_n(\cdot )} \left[ \exp \left( i \lambda n^{-1/2}\sum_{ \{j : S_j^n \in [-W, T] \}} \left(R^n_j -f_1(S_j^n) \right) \right) \right]  \nonumber\\
& \qquad = \prod_{ \{j : S_j^n \in [-W, T] \}} E^{S_j^n} \left[ \exp \left( i \lambda n^{-1/2} \left(\delta( S_j^n) - f_1(S^n_j) \right) \right) \right],
\end{align*}
which, using Lemma \ref{prod_approx_lemma}, can be written as
\begin{align}\label{expoapprox}
 &\prod_{ \{j : S_j^n \in [-W, T] \}}  E^{S_j^n} \left[ 1 - \lambda^2{(2n)}^{-1} {\left(\delta(S^n_j) - f_1(S_j^n)\right)}^2  \right]  + o_p(1)\nonumber \\
& \quad =   \prod_{ \{j : S_j^n \in [-W, T] \}}  \left[ 1 - \lambda^2{(2n)}^{-1} f_2(S_j^n) \right] + o_p(1) \nonumber \\
 & \quad =  \exp \left[ - \sum_{ \{j : S_j^n \in [-W, T] \}} -\log \left[ 1 - \lambda^2{(2n)}^{-1}f_2(S_j^n)  \right]\right] + o_p(1)
 \nonumber \\
 & \quad = \exp \left[ - \int_{[-W, T]} -\log \left( 1 - \lambda^2{(2n)}^{-1} f_2(x)\right) N^n(dx)  \right]  + o_p(1). 
  \end{align}
Step 2 : Now, since $|-\log(1-x) -x | \le 2 |x|^2$ if $|x| \le \frac{1}{2}$, we get that for large enough $n,$
\begin{align}\label{unifconv1}
\Big{|} -\log \left( 1 - \frac{\lambda^2}{2n}f_2(x)\right) - \frac{\lambda^2}{2n}f_2(x) \Big{|} &\le \frac{\lambda^4}{2n^2} f_2^2(x)  \le \frac{\lambda^4}{2n^2} {\left[ E[M^2([0, W])] \right]}^2 =  \frac{C}{n^2},
\end{align}
where $C = \frac{\lambda^4}{2}{\left[ E[M^2([0, W])] \right]}^2$. Hence, from Lemma \ref{use_sales_conv}, part \eqref{secondpart} and \eqref{unifconv1} it follows that
\begin{align*}
\left(\begin{array}{c}  \int_{-W}^T \left[ -\log \left( 1 - \frac{\lambda^2}{2n} f_2(x)\right) - \frac{\lambda^2}{2n} f_2(x) \right] N^n(dx) \\  \int_{-W}^T \frac{\lambda^2}{2n} f_2(x)N^n(dx) \end{array} \right) \stackrel{P}{\rightarrow} \left(\begin{array}{c} 0 \\ \int_{-W}^T \frac{\lambda^2}{2} f_2(x)\nu(dx) \end{array} \right), 
\end{align*}
and so,
\begin{align}\label{prelim_1}
\exp \left( -\int_{-W}^T  -\log \left( 1 - \frac{\lambda^2}{2n} f_2(x)\right) N^n(dx) \right) &\stackrel{P}{\rightarrow} \exp \left( -\int_{-W}^T \frac{\lambda^2}{2} f_2(x)\nu(dx) \right) \nonumber \\
& = \exp \left( - \frac{ \lambda^2 c_2}{2} \right).
\end{align}
Therefore, using \eqref{expoapprox}, \eqref{prelim_1} and the dominated convergence theorem, we get for $\lambda \in \R$,
\begin{align}\label{condn_conv}
E\left[ \Big{|} E^{N_n(\cdot )}\left[ \exp \left( i \lambda n^{-1/2}\left(R^n - \int_{[-W,T]} f_1(x) N^n(dx) \right)\right) \right] - \exp \left( - \frac{\lambda^2 c_2}{2} \right) \Big{|} \right] \rightarrow 0.
\end{align}

Step 3: Now, we prove the joint convergence of 
$$\left( \begin{array}{c} X^n\\ Y^n\end{array}\right) := \frac{1}{\sqrt{n}} \left( \begin{array}{c} R^n - \int_{[-W,T]} f_1(x) N^n(dx) \\  \int_{[-W,T]} f_1(x) N^n(dx) - nc \end{array}\right).$$ 

Consider the joint characteristic function $E\left[ \exp\left( i \lambda X^n + i \phi Y^n \right) \right]$ for $(\lambda, \phi) \in \R^2$ and note that
\begin{align}\label{joint_conv1}
E\left[ \exp\left( i \lambda X^n + i \phi Y^n \right) \right] & = E\left[ \exp\left( i \lambda X^n + i \phi Y^n \right) - \exp\left( -\frac{\lambda^2 c_2}{2} + i \phi Y^n \right) \right] \nonumber \\
& \quad + E\left[ \exp\left( -\frac{\lambda^2 c_2}{2} + i \phi Y^n \right) \right].
\end{align}
First, we deal with the first term on the right side of \eqref{joint_conv1}. Note that
\begin{align}\label{eqn:rn-first term}
\Big{|} E &\left[ \exp\left( i \lambda X^n + i \phi Y^n \right) - \exp\left( -\frac{\lambda^2 c_2}{2} + i \phi Y^n \right) \right] \Big{|}  \nonumber\\
&= \Big{|}E\left[ E^{N^n(\cdot)}\left[ \exp\left( i \lambda X^n + i \phi Y^n \right) - \exp\left( -\frac{\lambda^2 c_2}{2} + i \phi Y^n \right) \right] \right] \Big{|} \nonumber\\
& =  \Big{|} E\left[  \exp ( i \phi Y^n) \left[E^{N^n(\cdot)}\left[\exp\left( i \lambda X^n  \right) \right] - \exp\left( -\frac{\lambda^2 c_2}{2}  \right)  \right]  \right] \Big{|}  \nonumber \\
& \le E\left[ \Big{|} E^{N^n(\cdot)}\left[\exp\left( i \lambda X^n  \right) \right] - \exp\left( -\frac{\lambda^2 c_2}{2}  \right) \Big{|} \right]  \rightarrow 0,
\end{align}
where the last convergence follows from \eqref{condn_conv}. For the second term on the right side of \eqref{joint_conv1}, observe
\begin{equation*}
Y^n = \\ \left( \int_{[-W,T]} f_1(x) N^n(dx) - nc\right)/\sqrt{n} = \int_{[-W, T]} f_1(x) \left(\frac{ N^n - n \nu}{\sqrt{n}} \right) (dx).
\end{equation*}
Using Lemma \ref{use_sales_conv}, part \eqref{firstpart}, we get
\begin{align}\label{eqn:rn-secondterm}
\lim_{n \to \infty} E\left[ \exp\left( -\frac{\lambda^2 c_2}{2} + i \phi Y^n \right) \right] = \exp\left( -\frac{\lambda^2 c_2}{2}  + i \phi \tilde \mu - \frac{\phi^2 \tilde \sigma^2 }{2} \right), 
\end{align}
where the parameters $\tilde \mu$ and $\tilde \sigma^2$ are given in \eqref{tildemutildesigma}. 
Therefore, using \eqref{joint_conv1}, \eqref{eqn:rn-first term} and \eqref{eqn:rn-secondterm}, we get
\begin{align*}
\lim_{n \to \infty} E\left[ \exp\left( i \lambda X^n + i \phi Y^n \right) \right] = \lim_{n \to \infty} E\left[ \exp\left( -\frac{\lambda^2 c_2}{2} + i \phi Y^n \right) \right] = \exp\left( -\frac{\lambda^2 c_2}{2}  + i \phi \tilde \mu - \frac{\phi^2 \tilde \sigma^2 }{2} \right).
\end{align*}
So, the joint convergence holds, i.e. 
\begin{align*}
\left( \begin{array}{c} X^n \\ Y^n \end{array} \right) \Rightarrow \left( \begin{array}{c} X \\ Y \end{array} \right),
\end{align*}
on $\R^2,$ where $X$ and $Y$ are independent, $X \sim \mathcal{N}(0, c_2)$ and $Y \sim \mathcal{N}(\tilde \mu, \tilde \sigma^2)$. Hence, 
\begin{equation*}
\sqrt{n}\left( \frac{R^n}{n} - c_1 \right) = X^n + Y^n \Rightarrow X + Y, 
\end{equation*}
which gives us the required result.
\end{proof}

\begin{proof}[Proof of Theorem \ref{mainthm}]
Let $X_i$ denote the size of the $i$-th claim which arrived during the time interval $[0, T]$. Denote, $\s_j = \sum_{i =1}^j X_i$ for all $j \ge 1$. Then, $\cost([0, T]) = \sum_{i = 1}^{R^n} X_i = \s_{R^n}$.

(1) By Assumption \ref{size_iid} and Donsker's theorem \cite[page 146]{billingsley:1999}, we know
\begin{equation*}
\frac{\s_{[n \cdot]} - [ n \cdot ] E}{\sqrt{n V} }\Rightarrow W(\cdot)
\end{equation*}
on $D([0, \infty))$, where $W(\cdot)$ is a standard Brownian motion.
Also, from Theorem \ref{numberclaims}, we know that $\sqrt{n}\left( \frac{R^n}{n} - c_1 \right)$ converges in distribution to a random variable $Y$, where $Y \sim \mathcal{N}(\tilde \mu,c_2 + \tilde\sigma^2)$. These two facts, coupled with Assumption \ref{ind_amount_time} of Section \ref{freereplass} gives that
\begin{equation*}
\left( \begin{array}{c}\frac{\s_{[n \cdot]} - [ n \cdot ] E}{\sqrt{n V} } \\ \sqrt{n}\left( \frac{R^n}{n} - c_1 \right) \end{array} \right) \Rightarrow \left( \begin{array}{c} W(\cdot) \\ Y \end{array} \right)
\end{equation*} 
on $D([0, \infty)) \times \R,$ where $W(\cdot) $ and $Y$ are independent of each other. Hence, from \citet[page 37]{billingsley:1999}, we get
\begin{equation}\label{triple_conv}
\left( \begin{array}{c}\frac{\s_{[n \cdot]} - [ n \cdot ] E}{\sqrt{n V} } \\ \frac{ R^n}{n} \\ \sqrt{n}\left( \frac{R^n}{n} - c_1 \right) \end{array} \right) \Rightarrow \left( \begin{array}{c} W(\cdot) \\ c_1 \\ Y \end{array} \right)
\end{equation} 
on $D([0, \infty)) \times \R \times \R,$ where each three components on the right of \eqref{triple_conv} are independent of each other.
Applying Theorem 3 of \cite{durrett:resnick:1977} to \eqref{triple_conv} yields
\begin{equation}\label{northmeq1}
\left( \begin{array}{c}\frac{\s_{[R^n \cdot]} - [ R^n \cdot ] E}{\sqrt{n V} } \\ \sqrt{n}\left( \frac{R^n}{n} - c_1 \right) \end{array} \right) \Rightarrow \left( \begin{array}{c} W (c_1(\cdot)) \\ Y \end{array} \right)
\end{equation}
on $D([0, \infty)) \times \R.$ Since the first component of the limit in \eqref{northmeq1} is a continuous process, we have 
\begin{equation*}
\left( \begin{array}{c}\frac{\cost([0, T]) - R^n  E}{\sqrt{n V} } \\ \sqrt{n}\left( \frac{R^n}{n} - c_1 \right) \end{array} \right) \Rightarrow \left( \begin{array}{c} W(c_1) \\ Y \end{array} \right)
\end{equation*}
on $\R \times \R$ \cite[Section 4]{durrett:resnick:1977}. So,
\begin{equation*}
\frac{\cost([0,T]) - nc_1E}{\sqrt{n V} } = \frac{\cost([0,T]) - R^n E}{\sqrt{n V} } + \frac{E}{\sqrt{V}} \sqrt{n}\left( \frac{R^n}{n} - c_1 \right) \Rightarrow W(c_1) + \frac{E}{\sqrt{V}} Y
\end{equation*}
on $\R$, which completes the proof of part \eqref{normalthm} of Theorem \ref{mainthm}.

(2)
 By Assumption \ref{size_iid} and a minor variant of the central limit theorem of heavy-tailed distributions \cite[page 218]{resnickbook:2007}, we know
\begin{equation*}
\frac{\s_{[n \cdot]} - [ n \cdot ] E}{b(n) }\Rightarrow Z_{\alpha}(\cdot)
\end{equation*}
on $D([0, \infty))$, where $Z_{\alpha}(\cdot)$ is an $\alpha$-stable L\'{e}vy motion whose characteristic function satisfies \eqref{charstablemean}. Now, applying Theorem 3 of \cite{durrett:resnick:1977} in a similar way as in the proof of \eqref{normalapprox1}, we get
\begin{equation}\label{stablethmeq1}
\frac{\s_{[R^n \cdot]} - [ R^n \cdot ] E }{b(n) } \Rightarrow  Z_{\alpha}(c_1 (\cdot)),
\end{equation}
on $D([0, \infty)).$ Since the limit process in \eqref{stablethmeq1} is continuous in probability at time 1, we have 
\begin{equation}\label{stablealmostthere}
\frac{\cost([0, T]) - R^n  E}{b(n) }  \Rightarrow Z_{\alpha}(c_1) \stackrel{d}{=} c_1^{\frac{1}{\alpha}} Z_{\alpha}(1)
\end{equation}
on $\R$. Observe,
\begin{equation}\label{stablelaststep}
\frac{\cost([0,T]) - nc_1E}{b(n) } = \frac{\cost([0,T]) - R^n E}{b(n) } + \frac{R^nE -  n c_1 E }{b(n)}.
\end{equation}
In \eqref{stablealmostthere}, we have already found the weak limit of the first term on the right side of \eqref{stablelaststep}. Now, we deal with the second term on the right side of \eqref{stablelaststep}. Notice, $b(\cdot) \in RV_{\frac{1}{\alpha}}$, and hence $x^{1/2}{(b(x))}^{-1} \in RV_{ \frac{1}{2} -\frac{1}{\alpha}}$. Since $\alpha < 2$, $\frac{1}{2} -\frac{1}{\alpha} < 0$. Hence, $\lim_{ n \to \infty} \frac{\sqrt{n}}{ b(n)} = 0$. Therefore, using this fact, along with Theorem \ref{numberclaims} and Slutsky's theorem, we get
\begin{align*}
\frac{R^nE -  nc_1E }{b(n)} = \frac{\sqrt{n} E}{b(n)} \sqrt{n} \left( \frac{R^n}{n} - c_1 \right) \stackrel{P}{\rightarrow} 0,
\end{align*}
which together with \eqref{stablealmostthere} and \eqref{stablelaststep} completes the proof of part \eqref{stablemeanexists} of Theorem \ref{mainthm}.

(3)
 By Assumption \ref{size_iid} and the central limit theorem of heavy-tailed distributions \cite[page 218]{resnickbook:2007}, we know
\begin{equation*}
\frac{\s_{[n \cdot]} - [ n \cdot ] e(n)}{b(n) }\Rightarrow Z_{\alpha}(\cdot)
\end{equation*}
on $D([0, \infty)$), where $Z_{\alpha}(\cdot)$ is an $\alpha$-stable L\'{e}vy process and characteristic function of $Z_{\alpha}(c_1)$ satisfies \eqref{charstablenomean}. Now, applying Theorem 3 of \cite{durrett:resnick:1977} in a similar way as in the proof of \eqref{normalapprox1}, we get
\begin{equation}\label{stable2thmeq1nomean}
\frac{\s_{[R^n \cdot]} - [ R^n \cdot ] e(R^n)}{b(n) } \Rightarrow  Z_{\alpha}(c_1 (\cdot))
\end{equation}
on $D([0, \infty))$ and since the limit process in \eqref{stable2thmeq1nomean} is continuous in probability at time 1, we get
\begin{equation}\label{stablealmosttherenomean}
\frac{\cost([0, T]) - R^n  e(R^n)}{b(n) }  \Rightarrow Z_{\alpha}(c_1)
\end{equation}
on $\R$. Now, notice that
\begin{equation}\label{stablelaststepnomean}
\frac{\cost([0,T]) - nc_1^{\frac{1}{\alpha}}e(n)}{b(n) } = \frac{\cost([0,T]) - R^n e(R^n)}{b(n) } + \frac{R^n e(R^n) -  nc_1^{\frac{1}{\alpha}}e(n)}{b(n)}.
\end{equation}
We have already observed in \eqref{stablealmosttherenomean} the weak convergence of the first term on the right side of \eqref{stablelaststepnomean}. We now turn to the second term. To show the convergence in probability of the second term, we consider two separate cases, namely $\alpha =1$ and $\alpha < 1$. 

$\alpha < 1$ case: Notice $e(x) = I(b(x))$ where $I(x) = \int_0^x t F(dt) \in RV_{1 - \alpha}$ \citep[ page 36, Ex. 2.5]{resnickbook:2007}. Also, $b(x) \in RV_{\frac{1}{\alpha}}$, and $\lim_{x \to \infty} b(x) = \infty.$ Therefore, from Proposition 2.6(iv) of \citet[page 32]{resnickbook:2007}, $e(\cdot) \in RV_{\frac{1}{\alpha} -1}$. So, $xe(x) \in RV_{\frac{1}{\alpha}}$. Also, from Theorem \ref{numberclaims}, we know that $R^n/n$ converges in probability to $c_1$. Hence, \cite[page 36]{resnickbook:2007}
\begin{equation}\label{stablelaststep1nomean}
\frac{R^ne(R^n)}{ne(n)} \stackrel{P}{\rightarrow} c_1^{\frac{1}{\alpha}}.
\end{equation}
Using \eqref{stablelaststep1nomean} and the fact that the sequence $\{ \frac{ne(n)}{b(n)} \}$ converges to $\frac{\alpha}{1- \alpha}$ and hence is bounded, we get that
\begin{align*}
\frac{R^ne(R^n)-  nc_1^{\frac{1}{\alpha}}e(n)}{b(n)} = \frac{n e(n)}{b(n)} \left( \frac{R^ne(R^n)}{ne(n)} - c_1^{\frac{1}{\alpha}} \right) \stackrel{P}{\rightarrow} 0,
\end{align*}
which, together with \eqref{stablealmosttherenomean} and \eqref{stablelaststepnomean} completes the proof for the case $\alpha < 1$.

$\alpha = 1$ case: In this case, we may write the second term of the right side of \eqref{stablelaststepnomean} as
\begin{align}\label{alpha1nomean}
\frac{R^ne(R^n) -  nc_1e(n)}{b(n)} &= \frac{R^ne(R^n) -  R^ne(n)}{b(n)}  + \frac{R^n e(n) -  nc_1e(n)}{b(n)} \nonumber \\
&= \frac{R^n}{n}\frac{e(R^n) -  e(n)}{b(n)/n} + \frac{\sqrt{n}e(n)}{b(n)}\sqrt{n}\left( \frac{R^n}{n} - c_1 \right).
\end{align}

First we deal with the first term on the right side of \eqref{alpha1nomean}. Note that $\bar F(y)$ is (-1)-varying and so, using a theorem of \citet{dehaan:1976} (Proposition 0.11 of  \citet[page 30]{resnick:1987}), we get that $ \int_0^x  \bar F(y)dy$ is $\Pi$-varying with auxiliary function $x \bar F(x)$. Hence, using the fact that $I(x) = \int_0^x \bar F(y) dy - x\bar F(x),$ it easily follows from the definition of $\Pi$-varying functions \cite[page 27]{resnick:1987}, that $I(x)$ is also $\Pi$-varying with auxiliary function $x \bar F(x)$. Since, $b(x) \in RV_1$, $e(x) = I(b(x))$ is also $\Pi$-varying with auxiliary function $b(x)\bar F(b(x)) = b(x)/x$ \citep{dehaan:resnick:1979b}, \citep[page 38]{resnickbook:2007}. 

Using local uniform convergence of $\Pi$-varying functions on sets away from 0 \cite[page 139, Theorem 3.1.16]{bingham:goldie:teugels:1987}, and the fact $R^n/n \stackrel{P}{\rightarrow} c_1 > 0$, we get
\begin{align*}
\frac{e(R^n) -  e(n)}{b(n)/n} = \frac{e\left(n \cdot \frac{R^n}{n}\right) - e(n)}{b(n)/n}\stackrel{P}{\rightarrow} \log c_1.
\end{align*}
Hence,
\begin{align}\label{alpha1firstnomean}
\frac{R^ne(R^n) -  R^ne(n)}{b(n)} = \frac{R^n}{n}\frac{e(R^n) -  e(n)}{b(n)/n} \stackrel{P}{\rightarrow} c_1 \log c_1.
\end{align}

Now, we turn to the second term in \eqref{alpha1nomean}. Note that $b(n) \in RV_1$ and since $e(n)$ is $\Pi$-varying, $e(n) \in RV_0$ \citep[page 128]{bingham:goldie:teugels:1987}. Therefore, $\frac{\sqrt{n} e(n)}{ b(n)} \in RV_{ - \frac{1}{2} }$, and so, $\lim_{ n \to \infty} \frac{\sqrt{n}e(n)}{ b(n)} = 0.$ Also, from Theorem \ref{numberclaims}, we know that $\sqrt{n}\left(\frac{R^n}{n} - c_1 \right)$ is asymptotically normal. Therefore, using Slutsky's theorem, we get
\begin{align}\label{alpha1secnomean}
\frac{R^ne(n) -  nc_1e(n) }{b(n)} = \frac{\sqrt{n} e(n)}{b(n)} \sqrt{n} \left( \frac{R^n}{n} - c_1 \right) \stackrel{P}{\rightarrow} 0.
\end{align}

Finally, using \eqref{alpha1nomean}, \eqref{alpha1firstnomean} and \eqref{alpha1secnomean} we get
\begin{align*}
\frac{R^ne(R^n) -  nc_1e(n)}{b(n)} \stackrel{P}{\rightarrow} c_1 \log c_1,
\end{align*}
which, together with \eqref{stablealmosttherenomean} and \eqref{stablelaststepnomean} completes the proof for the case $\alpha =1$.
\end{proof}

\begin{proof}[Proof of Theorem \ref{proratathm}]

Denote by $\cost_j([0, T])$ the warranty claims expenditures which are incurred in $[0,T]$ for the $j$-th item, sold in the time interval $[-W, T]$. In notation,
\begin{equation*}
\cost_j([0, T]) = c_br(C^n_{j, 1})\epsilon_{C_{j, 1}^n}([0, W])\epsilon_{ S_j^n + C_{j, 1}^n} ([0, T]).
\end{equation*}
So, $\cost([0, T]) = \sum_{ \{j : S_j^n \in [-W, T] \}}  \cost_j([0, T])$. Using Assumption \ref{freerepl_claim1} of Section \ref{freereplass}, we get
\begin{align*}
{(c_b)}^{-1}\cost_j([0, T]) | S_j^n \stackrel{d}{=} \delta(S_j^n) . 
\end{align*}
Now, we consider the characteristic function of $$\left( {(c_b)}^{-1}\cost([0,T]) - \int_{[-W,T]} f_1(x)N^n(dx) \right)/ \sqrt{n}$$ given the sales process $N^n(\cdot)$. Note that for $\lambda \in \R$,
\begin{align*}
E^{N^n(\cdot)} &\left[ \exp \left( i \lambda n^{-1/2} \left( {(c_b)}^{-1}\cost([0,T]) - \int_{[-W,T]} f_1(x)N^n(dx) \right) \right)  \right] \nonumber \\
&\qquad = E^{N^n(\cdot)} \left[ \exp\left( i \lambda n^{-1/2} \sum_{ \{j : S_j^n \in [-W, T] \}}  \left( {(c_b)}^{-1}\cost_j([0, T]) - f_1(S_j^n) \right)\right) \right]   \nonumber \\
& \qquad = \prod_{ \{j : S_j^n \in [-W, T] \}} E^{S_j^n}\left[ \exp \left( i \lambda n^{-1/2} \left( \delta(S_j^n) - f_1(S_j^n) \right) \right) \right].
\end{align*}

Now, we want to use the same steps as in the proof of Theorem \ref{numberclaims}. To do so, we must be able to use Lemma \ref{prod_approx_lemma} and Lemma \ref{use_sales_conv}, which need Assumptions \ref{freerepl_sale}-\ref{freerepl_claim2} of Section \ref{freereplass}. By the hypothesis of Theorem \ref{proratathm}, we have Assumptions \ref{freerepl_sale}-\ref{freerepl_claim1.5} of Section \ref{freereplass}. Since $M([0,W]) \le 1$, we know that Assumption \ref{freerepl_claim2} of Section \ref{freereplass} is also satisfied. Hence, we can apply the results of Lemma \ref{prod_approx_lemma} and Lemma \ref{use_sales_conv} here, too. Now, the proof follows exactly similar steps as in the proof of Theorem \ref{numberclaims}. So, the rest of the proof is omitted.
\end{proof}

\section{Glossary of notation}\label{notation}

\begin{align*}
S_j^n &= \hbox{time of sale of the $j$-th item in the $n$-th model,}\\
N^n(t) &= \hbox{total number of sales in $[-W, t]$ in the $n$-th model} = \sum_j \epsilon_{S_j^n}([-W, t]),\\
\nu(t) &= \hbox{the centering function of $\frac{1}{n}N^n(t),$ }\\
 N^{\infty}(\cdot) &= \hbox{the Gaussian process limit of centered and scaled $\frac{1}{n}N^n(\cdot)$, }\\
C_{j, i}^n &= \hbox{the time of the $i$-th claim for the $j$-th item sold, where we start our clock at the time} \\ & \qquad \hbox{of sale $S_j^n$,}\\
M_j^n(\cdot) &:= \sum_i \epsilon_{C_{j, i}^n} (\cdot) = \hbox{the times of claims measure for the $j$-th item sold},\\
M(\cdot) &= \hbox{the generic random measure representing times of claims for an item sold, i.e.}\\
& \quad M(\cdot) \stackrel{d}{=} M_j^n(\cdot) \hbox{ for all $j$ and all $n$.}\\ 
m(\cdot) &= E[M(\cdot)],\\
\hat m_1(\cdot) &= \frac{1}{n} \sum_{j=1}^n M_j^n(\cdot),\\
\hat m(\cdot) &= \hbox{estimate of $m(\cdot)$, maybe using a functional form},\\
R_j^n &= \hbox{ total number of claims in $[0, T]$ for the $j$-th item sold} =  \sum_i \epsilon_{ S_j^n + C_{j, i}^n} ([0, T])\epsilon_{C_{j, i}^n} ([0, W]),\\
R^n &= \hbox{ total number of claims in $[0, T]$} =  \sum_{\{j : S_j^n \in [-W, T] \}} R_j^n,\\
r(t) &= \left \{ \begin{array}{ll}  \hbox{$1,$} & \hbox{if we follow the non-renewing} \\
& \hbox{free replacement policy,}\\
\\
\hbox{the fraction of the cost refunded  if the claim} & \hbox{if we follow the non-renewing } \\ 
\hbox{comes after time t from the date of sale,} & \hbox{pro-rata policy,} \end{array} \right.\\
 \tilde m(\cdot) &= E\left[\int_{(\cdot)} r(y)M(dy) \right],\\
\delta(x) &= \left \{\begin{array}{ll}  \int_{[0, T-x] } r(y) M(dy), & \hbox{if $ 0 \le x \le T$},\\
 \int_{[-x, T-x] } r(y) M(dy), & \hbox{if $ T - W < x < 0$},\\
  \int_{[-x, W] } r(y) M(dy), & \hbox{if $ -W \le x \le T -W,$} \end{array} \right.\\
   f_1(x) & = E [\delta(x)]   =  \left \{\begin{array}{ll} \tilde m([0, T-x]), & \hbox{if $ 0 \le x \le T$},\\
 \tilde m([-x, T-x]), & \hbox{if $ T - W < x < 0$},\\
  \tilde m([-x, W]), & \hbox{if $ -W \le x \le T -W$},\end{array} \right. \\
  f_2(x) &= Var[ \delta(x)],\\
  c_1 &= \int_{-W}^T f_1(x)\nu(dx), \\
  c_2 &= \int_{-W}^T f_2(x)\nu(dx),\\
  E &= \hbox{expectation of the distribution of claim sizes, if it exists},\\
  V &= \hbox{variance of the distribution of claim sizes, if it exists},\\
  &\hbox{$\chi : D([-W, T]) \rightarrow \R^{[0,W]}$ defined as}\\
  \chi(x)(u) &= x(T-u) - x((-u)-), \\
\tilde \mu &= \int_{[0, W]} E[ \chi \left(N^{\infty} \right)(u)] \tilde m(du),\\
 \tilde \sigma^2 &= \int_{[0, W]} \int_{[0, W]} Cov \left[ \chi \left(N^{\infty} \right)(u), \chi \left(N^{\infty} \right)(v) \right] \tilde m(du)\tilde m(dv).
 \end{align*}

\bibliographystyle{plainnat}
\bibliography{bibfile}
  \end{document}